\documentclass[english]{article}
\usepackage[T1]{fontenc}
\usepackage[latin9]{inputenc}
\usepackage{geometry}
\usepackage{float}
\usepackage{amsmath}
\usepackage{amsthm}
\usepackage{amssymb}
\usepackage{graphicx}
\usepackage{esint}
\usepackage[numbers]{natbib}

\makeatletter

\providecommand{\tabularnewline}{\\}
\floatstyle{ruled}
\newfloat{algorithm}{tbp}{loa}
\providecommand{\algorithmname}{Algorithm}
\floatname{algorithm}{\protect\algorithmname}

\numberwithin{equation}{section}
\numberwithin{figure}{section}
\theoremstyle{plain}
\newtheorem{thm}{\protect\theoremname}
  \theoremstyle{plain}
  \newtheorem{lem}[thm]{\protect\lemmaname}
  \theoremstyle{remark}
  \newtheorem{rem}[thm]{\protect\remarkname}
  \theoremstyle{definition}
  \newtheorem{defn}[thm]{\protect\definitionname}
  \theoremstyle{plain}
  \newtheorem{cor}[thm]{\protect\corollaryname}
  \theoremstyle{plain}
  \newtheorem{prop}[thm]{\protect\propositionname}

\pdfoutput=1

\usepackage{amsmath, amsthm, amssymb, amsfonts}
\usepackage{bbm}
\usepackage{bbold}
\usepackage[basic]{complexity}
\usepackage[pdftex,bookmarks,colorlinks]{hyperref}
\usepackage{enumitem}
\usepackage{color}

\usepackage[lined,boxed,ruled,norelsize,algo2e]{algorithm2e}

\makeatother

\usepackage{babel}
  \providecommand{\corollaryname}{Corollary}
  \providecommand{\definitionname}{Definition}
  \providecommand{\lemmaname}{Lemma}
  \providecommand{\propositionname}{Proposition}
  \providecommand{\remarkname}{Remark}
\providecommand{\theoremname}{Theorem}

\begin{document}
\global\long\def\R{\mathbb{R}}
 \global\long\def\Z{\mathbb{Z}}

\global\long\def\ellOne{\ell_{1}}
 \global\long\def\ellTwo{\ell_{2}}
 \global\long\def\ellInf{\ell_{\infty}}

\global\long\def\boldVar#1{\mathbf{#1}}
\global\long\def\mvar#1{\boldVar{#1}}
\global\long\def\vvar#1{\vec{#1}}


\global\long\def\defeq{\stackrel{\mathrm{{\scriptscriptstyle def}}}{=}}
\global\long\def\E{\mathbb{E}}
\global\long\def\otilde{\tilde{O}}

\global\long\def\diag{\boldsymbol{diag}}


\global\long\def\gradient{\bigtriangledown}
 \global\long\def\grad{\gradient}
 \global\long\def\hessian{\gradient^{2}}
 \global\long\def\hess{\hessian}
 \global\long\def\jacobian{\mvar J}

 \global\long\def\setVec#1{\onesVec_{#1}}
 \global\long\def\indicVec#1{\onesVec_{#1}}

\global\long\def\specGeq{\succeq}
 \global\long\def\specLeq{\preceq}
 \global\long\def\specGt{\succ}
 \global\long\def\specLt{\prec}

\global\long\def\innerProduct#1#2{\big\langle#1 , #2 \big\rangle}
 \global\long\def\norm#1{\left\Vert #1\right\Vert }
\global\long\def\normFull#1{\left\Vert #1\right\Vert }

\global\long\def\opt{\mathrm{opt}}

\global\long\def\fopt{f^{*}}

\global\long\def\va{\vvar a}
 \global\long\def\vb{\vvar b}
 \global\long\def\vc{\vvar c}
 \global\long\def\vd{\vvar d}
 \global\long\def\ve{\vvar e}
 \global\long\def\vf{\vvar f}
 \global\long\def\vg{\vvar g}
 \global\long\def\vh{\vvar h}
 \global\long\def\vl{\vvar l}
 \global\long\def\vm{\vvar m}
 \global\long\def\vn{\vvar n}
 \global\long\def\vo{\vvar o}
 \global\long\def\vp{\vvar p}
 \global\long\def\vq{\vvar q}
 \global\long\def\vr{\vvar r}
 \global\long\def\vs{\vvar s}
 \global\long\def\vu{\vvar u}
 \global\long\def\vv{\vvar v}
 \global\long\def\vw{\vvar w}
 \global\long\def\vx{\vvar x}
 \global\long\def\vy{\vvar y}
 \global\long\def\vz{\vvar z}

\global\long\def\vpi{\vvar{\pi}}
\global\long\def\vxi{\vvar{\xi}}
\global\long\def\vchi{\vvar{\chi}}
 \global\long\def\valpha{\vvar{\alpha}}
 \global\long\def\veta{\vvar{\eta}}
 \global\long\def\vlambda{\vvar{\lambda}}
 \global\long\def\vmu{\vvar{\mu}}
\global\long\def\vdelta{\vvar{\Delta}}
 \global\long\def\vsigma{\vvar{\sigma}}
 \global\long\def\vzero{\vvar 0}
 \global\long\def\vones{\vvar 1}

\global\long\def\xopt{\vvar x^{*}}

\global\long\def\ma{\mvar A}
 \global\long\def\mb{\mvar B}
 \global\long\def\mc{\mvar C}
 \global\long\def\md{\mvar D}
\global\long\def\mE{\mvar E}
 \global\long\def\mf{\mvar F}
 \global\long\def\mg{\mvar G}
 \global\long\def\mh{\mvar H}
\global\long\def\mI{\mvar I}
 \global\long\def\mm{\mvar M}
 \global\long\def\mn{\mathbf{N}}
\global\long\def\mq{\mvar Q}
 \global\long\def\mr{\mvar R}
 \global\long\def\ms{\mvar S}
 \global\long\def\mt{\mvar T}
 \global\long\def\mU{\mvar U}
 \global\long\def\mv{\mvar V}
 \global\long\def\mw{\mvar W}
 \global\long\def\mx{\mvar X}
 \global\long\def\my{\mvar Y}
\global\long\def\mz{\mvar Z}
 \global\long\def\mproj{\mvar P}
 \global\long\def\mSigma{\mvar{\Sigma}}
 \global\long\def\mLambda{\mvar{\Lambda}}
 \global\long\def\mha{\hat{\mvar A}}
 \global\long\def\mzero{\mvar 0}
\global\long\def\mlap{\mvar{\mathcal{L}}}
\global\long\def\mpi{\mvar{\mathcal{\Pi}}}

\global\long\def\mdiag{\mvar{\texttt{d}iag}}

\global\long\def\weightVec{\vvar w}
  \global\long\def\oracle{\mathcal{O}}
 \global\long\def\moracle{\mvar O}
 \global\long\def\oracleOf#1{\oracle\left(#1\right)}
 \global\long\def\nSamples{s}
 \global\long\def\simplex{\Delta}

\global\long\def\abs#1{\left|#1\right|}

\global\long\def\capacityMatrix{\mvar U}

\global\long\def\cost{\mathrm{cost}}
\global\long\def\tr{\mathrm{tr}}

\global\long\def\timeNearlyOp{\tilde{\mathcal{O}}}
 \global\long\def\timeNearlyLinear{\timeNearlyOp}
\global\long\def\elf{\hat{f}}
\global\long\def\elfle{\hat{f}_{e}^{(L)}}
\global\long\def\elfre{\hat{f}_{e}^{(R)}}
\global\long\def\elfl{\hat{f}^{(L)}}
\global\long\def\elfr{\hat{f}^{(R)}}
\global\long\def\potdif{\Delta\hat{\phi}}
\global\long\def\potr{\hat{\phi}^{(R)}}
\global\long\def\potl{\hat{\phi}^{(L)}}
\global\long\def\pot{\hat{\phi}}
\global\long\def\deg{\textnormal{deg}}
 \global\long\def\degout{\textnormal{deg}_{\textnormal{out}}}
\global\long\def\degin{\textnormal{deg}_{\textnormal{in}}}
\global\long\def\deginG#1{\deg_{in}^{#1}}
\global\long\def\degoutG#1{\deg_{out}^{#1}}
\global\long\def\leftg{Left[G]}
\global\long\def\rightg{Right[G]}

\global\long\def\elfpe{\hat{f}_{e}^{(P)}}
\global\long\def\elfqe{\hat{f}_{e}^{(Q)}}
\global\long\def\elfp{\hat{f}^{(P)}}
\global\long\def\elfq{\hat{f}^{(Q)}}
\global\long\def\rhop{\rho^{(P)}}
\global\long\def\rhoq{\rho^{(Q)}}
\global\long\def\fp{f^{(P)}}
\global\long\def\elfn{\hat{f}^{(N)}}
\global\long\def\fn{f^{(N)}}
\global\long\def\potp{\hat{\phi}^{(P)}}
\global\long\def\potq{\hat{\phi}^{(Q)}}
\global\long\def\potn{\hat{\phi}^{(N)}}

\global\long\def\crho{c_{\rho}}
\global\long\def\cT{c_{T}}

\global\long\def\varFun{f}

\global\long\def\funLip{L}
 \global\long\def\funCon{\mu}

\global\long\def\stepOpt{T}
 \global\long\def\stepOptCoordinate#1{\stepOpt_{(#1)}}

\global\long\def\reff#1{R_{#1}^{\text{eff}}}

\global\long\def\im{\mathrm{im}}

\global\long\def\ceil#1{\left\lceil #1 \right\rceil }

\global\long\def\runtime{\mathcal{T}}
 \global\long\def\timeOf#1{\runtime\left(#1\right)}

\global\long\def\domain{\mathcal{D}}

\global\long\def\argmin{\mathrm{argmin}}
\global\long\def\argmax{\mathrm{argmax}}
\global\long\def\nnz{\mathrm{nnz}}
\global\long\def\vol{\mathrm{vol}}
\global\long\def\supp{\mathrm{supp}}
\global\long\def\dist{\mathcal{D}}

\global\long\def\ora#1{\overrightarrow{#1}}

\global\long\def\oG{\bar{G}}

\global\long\def\energy#1#2{\mathcal{E}_{#1}\left(#2\right)}

\global\long\def\vphi{\mathit{\phi}}

\global\long\def\tvphi{\tilde{\phi}}
\global\long\def\pinv#1{#1^{\dagger}}

\title{Negative-Weight Shortest Paths and Unit Capacity Minimum Cost Flow
in $\otilde\left(m^{10/7}\log W\right)$ Time}

\author{%
\begin{tabular}{cccc}
Michael B. Cohen & Aleksander M\k{a}dry & Piotr Sankowski & Adrian Vladu \tabularnewline
MIT & MIT & University of Warsaw & MIT\tabularnewline
micohen@mit.edu & madry@mit.edu & sank@mimuw.edu.pl & avladu@mit.edu\tabularnewline
\end{tabular} }
\maketitle
\begin{abstract}
In this paper, we study a set of combinatorial optimization problems
on weighted graphs: the shortest path problem with negative weights,
the weighted perfect bipartite matching problem, the unit-capacity
minimum-cost maximum flow problem and the weighted perfect bipartite
$b$-matching problem under the assumption that $\norm b_{1}=O(m)$.
We show that each one of these four problems can be solved in $\tilde{O}(m^{10/7}\log W)$
time, where $W$ is the absolute maximum weight of an edge in the
graph, which gives the first in over 25 years polynomial improvement
in their sparse-graph time complexity. 

At a high level, our algorithms build on the interior-point method-based
framework developed by M\k{a}dry (FOCS 2013) for solving unit-capacity
maximum flow problem. We develop a refined way to analyze this framework,
as well as provide new variants of the underlying preconditioning
and perturbation techniques. Consequently, we are able to extend the
whole interior-point method-based approach to make it applicable in
the weighted graph regime. 
\end{abstract}

\section{Introduction}

In 2013, M\k{a}dry \cite{Madry13} put forth an algorithm for the
maximum flow and maximum-cardinality bipartite matching problems that
improved over a long standing $O(n^{3/2})$ running time barrier for
sparse graphs. Specifically, he presented an $\tilde{O}(m^{10/7})$
time algorithm for computing maximum flow in a unit capacity network
-- which implies an $\tilde{O}(m^{10/7})$ time algorithm for the
maximum-cardinality bipartite matching problem as well. The core of
his approach is a new path-following interior-point type algorithm
for solving a certain ``bipartite $b$-matching'' problem. At a
high level, this algorithm uses electrical flow computations to improve
the maintained primal dual solution and move it along the so-called
central path, i.e., a set of primal dual solutions in which every
edge contributes approximately the same amount to the duality gap.
As M\k{a}dry has shown, one can use this framework to establish an
$O(m^{3/7})$ bound on the number of such electrical flow computations
needed to compute a (near-) optimal solution to the considered problem
and thus to improve upon the generic worst-case $O(\sqrt{m})$ bound
that all the previous interior-point-method-based algorithms provided.
The key ingredient needed to obtaining this improved bound was a technique
for perturbing and preconditioning the intermediate solutions that
emerge during the computations. Unfortunately, the technique \cite{Madry13}
used in was inherently unable to cope with large capacity or \emph{any}
weights on edges. In fact, it did not provide any meaningful result
for the unit-capacity minimum cost maximum flow problem even when
all the edge weights were equal to $1$. Consequently, it remains
an open question whether a similar running time improvement can be
achieved for either: (a) non-unit capacity networks; or (b) weighted
variants of the unit-capacity graph problems.

\subsection{Our Contribution}

In this paper, we answer the second question above affirmatively by
providing an $\tilde{O}(m^{10/7}\log W)$ time algorithm for the minimum
cost unit-capacity maximum flow problem. In addition to the improvement
for this fundamental graph problem, this result also improves several
other standard problems as immediate corollaries. Namely, by well-known
reductions, it implies $\tilde{O}(m^{10/7}\log W)$ time algorithms
for the minimum-weight bipartite perfect matching problem, the\textbf{
}minimum-weight bipartite $b$-matching problem, and, the shortest
path problem for graphs with negative weights. This constitutes the
first in more than 25 years polynomial improvement of sparse graph
time complexity for each one of these problems.

To obtain these results we simplify and extend the framework from
\cite{Madry13} by developing new preconditioning and perturbation
techniques. These techniques provide us with much better control of
the intermediate solutions that we precondition/perturb. In particular,
in stark contrast to \cite{Madry13}, the preconditioning and perturbation
steps do not lead to any changes in edge costs. Also, the resulting
changes in vertex demands are very minimal. Thus, there is no more
need for repeated fixing of these vertex demand changes throughout
the execution of the algorithm -- a single demand correction step
is performed only at the very end. Finally, our analysis of the resulting
algorithm is much more streamlined and principled compared to the
analysis performed in \cite{Madry13}, providing us with much tighter
grip of the trade-offs underlying the whole framework.

\subsection{Previous Work}

The minimum-cost flow, min-weight bipartite perfect matching as well
as the shortest path with negative weights problems are core combinatorial
optimization tasks that now have been studied for over 85 years, starting
with the work of Egerváry \cite{egervary31} from 1931. Due to immense
number of works on these topics we will not review them. Instead we
will concentrate only on the ones that are relevant to the sparse
graph case, as that is the regime where our results are of main importance.

\paragraph{Shortest Paths with Negative Weights}

\begin{table}[H]
\begin{centering}
\begin{tabular}{c|c|c|}
\cline{2-3} 
 & Complexity & Author\tabularnewline
\cline{2-3} 
 & {\small{}$O(n^{4})$} & {\small{}Shimbel (1955) \cite{shimbel55}}\tabularnewline
\cline{2-3} 
 & $O(Wn^{2}m)$ & {\small{}Ford (1956) \cite{ford56} }\tabularnewline
\cline{2-3} 
{*} & {\small{}$O(nm)$} & B{\small{}ellman (1958) \cite{bellman58}, Moore (1959) \cite{moore59} }\tabularnewline
\cline{2-3} 
 & {\small{}$O(n^{\frac{3}{4}}m\log W)$} & G{\small{}abow (1983) \cite{gabow-85}}\tabularnewline
\cline{2-3} 
 & {\small{}$O(\sqrt{n}m\log(nW))$} & G{\small{}abow and Tarjan (1989) \cite{gabow-tarjan-89} }\tabularnewline
\cline{2-3} 
{*} & {\small{}$O(\sqrt{n}m\log(W))$ } & {\small{}Goldberg (1993) \cite{goldberg93} }\tabularnewline
\cline{2-3} 
{*} & {\small{}$\tilde{O}(Wn^{\omega})$} & {\small{}Sankowski (2005) \cite{s05} Yuster and Zwick (2005) \cite{yz05}}\tabularnewline
\cline{2-3} 
{*} & {\small{}$\tilde{O}(m^{10/7}\log W)$} & this paper\tabularnewline
\cline{2-3} 
\end{tabular}
\par\end{centering}

\caption{The complexity results for the SSSP problem with negative weights
({*} indicates asymptotically the best bound for some range of parameters). }
\end{table}

A list of the complexity results on single source shortest paths with
negative weights is included in Table 1. Observe that the sparse case
can either be solved in $O(mn)$ time \cite{bellman58,moore59} or
$\tilde{O}(\sqrt{n}m\log$W) time \cite{gabow-tarjan-89,goldberg93}.
The only progress that we had since these papers was the reduction
of the problem to fast matrix multiplication \cite{s05,yz05} that
is relevant only for dense graphs with small integral weights.

\paragraph{Min-cost Perfect Bipartite Matching}

\begin{table}[H]
\begin{centering}
\begin{tabular}{c|c|c|}
\cline{2-3} 
 & Complexity & Autor\tabularnewline
\cline{2-3} 
 & $O(Wn^{2}m)$ & Egerváry (1931) \cite{egervary31}\tabularnewline
\cline{2-3} 
 & $O(n^{4})$ & Khun (1955) \cite{kuhn55} and Munkers (1957) \cite{munkres57}\tabularnewline
\cline{2-3} 
 & $O(n^{2}m)$ & Iri (1960) \cite{iri60} \tabularnewline
\cline{2-3} 
 & $O(n^{3})$ & Dinic and Kronrod (1969) \cite{dinic-kronrod-69} \tabularnewline
\cline{2-3} 
{*} & $O(nm+n^{2}\log n)$ & Edmonds and Karp (1970) \cite{edmonds-karp-72} \tabularnewline
\cline{2-3} 
 & $O(n^{\frac{3}{4}}m\log W)$ & Gabow (1983) \cite{gabow-85} \tabularnewline
\cline{2-3} 
{*} & $O(\sqrt{n}m\log(nW))$ & Gabow and Tarjan (1989) \cite{gabow-tarjan-89} \tabularnewline
\cline{2-3} 
{*} & $O(W\sqrt{n}m)$ & Kao, Lam, Sung and Ting (1999) \cite{kao-99}\tabularnewline
\cline{2-3} 
{*} & $O(Wn^{\omega})$ & Sankowski (2006) \cite{s06} \tabularnewline
\cline{2-3} 
{*} & $\tilde{O}(m^{10/7}\log W)$ & this paper\tabularnewline
\cline{2-3} 
\end{tabular}
\par\end{centering}

\caption{The complexity results for the minimum weight bipartite perfect matching
problem ({*} indicates asymptotically the best bound for some range
of parameters). }
\end{table}

The complexity survey of for the minimum weight bipartite perfect
matching problem is given in Table 2. Here, the situation is very
similar to the case of shortest paths. We have two results that are
relevant for the sparse case considered here: $O(nm+n^{2}\log n)$
time \cite{edmonds-karp-72} and $O(\sqrt{n}m\log(nW))$ time \cite{gabow-tarjan-89}.
Again the only polynomial improvement that was achieved during the
last 25 years is relevant to the dense case only \cite{s06}.

\paragraph{Minimum-cost Unit-capacity Maximum Flow Problem}

\begin{table}[H]
\begin{centering}
\begin{tabular}{c|c|c|}
\cline{2-3} 
 & Complexity & Author\tabularnewline
\cline{2-3} 
 & $O(m(m+n\log n)$ & Edmonds and Karp (1972) \cite{edmonds-karp-72}\tabularnewline
\cline{2-3} 
 & $O(n^{5/3}m^{2/3}\log(nW)$) & Goldberg and Tarjan (1987) \cite{Goldberg:1987}\tabularnewline
\cline{2-3} 
 & $O(\min(\sqrt{m},n^{2/3})m\log(nW))$ & Gabow and Tarjan (1989) \cite{gabow-tarjan-89}\tabularnewline
\cline{2-3} 
 & $\tilde{O}(m^{3/2})$ & Daitch and Spielman (2008) \cite{Daitch:2008}\tabularnewline
\cline{2-3} 
 & $\tilde{O}(\sqrt{n}m)$ & Lee and Sidford (2014) \cite{lee-sidford}\tabularnewline
\cline{2-3} 
 & $\tilde{O}(m^{10/7}\log W)$ & this paper\tabularnewline
\cline{2-3} 
\end{tabular}
\par\end{centering}

\caption{The summary of the results for the min-cost unit-capacity max-flow
problem. For simplicity we only list exact algorithms that yielded
polynomial improvement or new strongly polynomial bounds. For the
full list of complexities we refer the reader to Chapter 12 in \cite{schrijver-book}. }
\end{table}

Due to the vast number of results for this topic we restricted ourselves
to present in Table 3 only algorithms for the unit-capacity case that
yielded significant improvement. We note, however, that handling general
capacities for this problem is a major challenge and our restriction
to unit-capacity networks oversimplifies the history of this problem.
Nevertheless, for there sparse case there are two relevant complexities:
a $\tilde{O}(\sqrt{n}m)$ time bound of \cite{lee-sidford} that,
for the sparse graph case, matches the previously best known bound
$O(m^{3/2}\log(nW))$ due to \cite{gabow-tarjan-89}. We note that
there was also a limited progress \cite{6686149} on fast matrix multiplication
based algorithms for the small vertex capacity variant of this problem.

\paragraph{Minimum-cost Perfect Bipartite $\boldsymbol{b}$-Matching}

\begin{table}[H]
\begin{centering}
\begin{tabular}{c|c|c|}
\cline{2-3} 
 & Complexity & Author\tabularnewline
\cline{2-3} 
 & $O(m(m+n\log n)$ & Edmonds and Karp (1972) \cite{edmonds-karp-72}\tabularnewline
\cline{2-3} 
 & $O(m^{7/4}\log W)$ & Gabow (1985) \cite{gabow-85}\tabularnewline
\cline{2-3} 
 & $O(m^{3/2}\log(nW))$ & Gabow and Tarjan (1989) \cite{gabow-tarjan-89}\tabularnewline
\cline{2-3} 
 & $\tilde{O}(m^{10/7}\log W)$ & this paper\tabularnewline
\cline{2-3} 
\end{tabular}
\par\end{centering}

\caption{The summary of the results for the min-cost perfect bipartite $b$-matching
problem under assumption that $b(V)=O(m).$ For simplicity we only
list exact algorithms that yielded polynomial improvement or new strongly
polynomial bounds. For the full list of complexities we refer the
reader to Chapter 21 in \cite{schrijver-book}. }
\end{table}

Minimum-cost perfect bipartite $b$-matching problem bears many similarities
to the minimum-cost maximum flow problem, e.g., see our reduction
of min-cost flow to $b$-matchings in Section \ref{sec:reduction}.
Hence some of the complexities in Table 4 are the same as in Table
3. However, not all results seem to carry over as exemplified in the
tables. $b$-matchings seems slightly harder than max-flow as one
needs to route exact amount of flow between many sources and sinks.
The results relevant for sparse case are: weakly polynomial $O(m^{3/2}\log(nW))$
time algorithm \cite{gabow-tarjan-89} or strongly polynomial $O(m(m+n\log n))$
time algorithm \cite{edmonds-karp-72}.

\subsection{The Outline of Our Algorithm}

As mentioned above, our focus is on development of a faster, $\tilde{O}(m^{10/7}\log W)$-time
algorithm for the minimum-cost unit-capacity maximum flow problem,
since well-known reductions immediately yield $\tilde{O}(m^{10/7}\log W)$-time
algorithms for the remaining problems (see also Section \ref{sec:Shortest-Paths-with}).
In broad outline, our approach to solving that flow problem follows
the framework introduced by M\k{a}dry \cite{Madry13} and comprises
three major components. 

First, in Section \ref{sec:reduction}, we reduce our input minimum-cost
flow instance to an instance of the bipartite minimum-cost $b$-matching
problem. The latter instance will have a special structure. In particular,
it can be viewed as a minimum-cost flow problem instance that has
general flow demands but no capacities. 

Then, in Section \ref{sec:basic_framework}, we put forth a basic
interior-point method framework that builds on the framework of M\k{a}dry
\cite{Madry13} for solving this kind of uncapacitated minimum-cost
flow instances. This basic framework alone will not be capable of
breaking the $O(\sqrt{m})$ iteration bound. Therefore, in Sections
\ref{sec:longer_steps} and \ref{sec:preconditioning}, we develop
a careful perturbation and preconditioning technique to help us control
the convergence behavior of our interior-point method framework. An
important feature of this technique is that, in contrast to the technique
used by \cite{Madry13}, our perturbations do no alter the arc costs.
Thus they are suitable for dealing with weighted problems. We then
prove that the resulting algorithms indeed obtains the improved running
time bound of $\tilde{O}(m^{10/7}\log W)$ but the (near-)optimal
flow solution it outputs might have some of the flow demands changed.

Finally, in Section \ref{sec:repair}, we address this problem by
developing a fast procedure that relies on classic combinatorial techniques
and recovers from this perturbed (near-)optimal solution an optimal
solution to our original minimum-cost flow instance.

\section{Preliminaries}

In this section, we introduce some basic notation and definitions
that we will need later.

\subsection{Minimum Cost $\sigma$-Flow}

We denote by $G=(V,E,c)$ a directed graph with vertex set $V$, arc
set $E$ and cost function $c$. We denote by $m=\abs E$ the number
of arcs in $G$, and by $n=\abs V$ its number of vertices. An arc
$e$ of $G$ connects an ordered pair $(w,v)$, where $w$ is called
\emph{tail} and $v$ is called \emph{head}. We will be mostly working
with \emph{$\sigma$-flows} in $G$, where $\sigma\in\mathbb{R}^{n}$,
satisfying $\sum_{v}\sigma_{v}=0$, is the \emph{demand vector}. A
$\sigma$-flow in $G$ is defined to be a vector $f\in\mathbb{R}^{m}$
(assigning values to arcs of $G$) that satisfies the following \emph{flow
conservation constraint}s: 
\begin{equation}
\sum_{e\in E^{+}(v)}f_{e}-\sum_{e\in E^{-}(v)}f_{e}=\sigma_{v},\quad\text{for each vertex \ensuremath{v\in V}}.\label{eq:conservation_constraints}
\end{equation}
We denote by $E^{+}(v)$ (resp. $E^{-}(v)$) the set of arcs of $G$
that are leaving (resp. entering) vertex $v$. The above constraints
enforce that, for every $v\in V$, the total out-flow from $v$ minus
the total in-flow (i.e. the net flow) out of $v$ is equal to $\sigma_{v}$.

Furthermore, we say that a $\sigma$-flow $f$ is \emph{feasible}
in $G$ iff it satisfies \emph{non-negativity and capacity constraints}
(in this paper we are only concerned with unit capacities): 
\begin{equation}
0\leq f_{e}\leq1,\quad\text{for each arc \ensuremath{e\in E}}.\label{eq:capacity_constraints}
\end{equation}

For our interior point method the basic object is the \emph{minimum
cost }$\sigma$\emph{-flow problem} consists of finding feasible $\sigma$-flow
$f$ that minimizes \emph{the cost of the flow} $c(f)=\sum_{e}c_{e}f_{e}$.

\subsection{Electrical Flows and Potentials}

A notion that will play a fundamental role in this paper is the notion
of electrical flows. Here, we just briefly review some of the key
properties that we will need later. For an in-depth treatment we refer
the reader to \cite{Bollobas98}.

Consider an undirected graph $G$ and a vector of resistances $r\in\mathbb{R}^{m}$
that assigns to each edge $e$ its \emph{resistance} $r_{e}>0$. For
a given $\sigma$-flow $f$ in $G$, let us define its \emph{energy}
(with respect to $r$) $\energy rf$ to be 
\begin{equation}
\energy rf:=\sum_{e}r_{e}f_{e}^{2}=f^{\top}Rf,\label{eq:def_energy_flow}
\end{equation}
where $R=\textnormal{diag}(r)$ is an $m\times m$ diagonal matrix
with $R_{e,e}=r_{e}$, for each edge $e$. In order to simplify notation,
we drop the subscript or the parameter whenever it is clear from the
context.

For a given undirected graph $G$, a demand vector $\sigma$, and
a vector of resistances $r$, we define an \emph{electrical $\sigma$-flow}
in $G$ (that is \emph{determined} by the resistances $r$) to be
the $\sigma$-flow that minimizes the energy $\energy rf$ among all
$\sigma$-flows in $G$. As energy is a strictly convex function,
one can easily see that such a flow is unique. Also, we emphasize
that we do \emph{not} require here that this flow is feasible with
respect to the (unit) capacities of $G$ (cf. (\ref{eq:capacity_constraints})).
Furthermore, whenever we consider electrical flows in the context
of a directed graph $G$, we will mean an electrical flow -- as defined
above -- in the (undirected) projection $\oG$ of $G$.

One of very useful properties of electrical flows is that it can be
characterized in terms of vertex potentials inducing it. Namely, one
can show that a $\sigma$-flow $f$ in $G$ is an electrical $\sigma$-flow
determined by resistances $r$ iff there exist \emph{vertex potentials}
$\phi_{v}$ (that we collect into a vector $\vphi\in\mathfrak{\mathbb{R}}^{n}$)
such that, for any edge $e=(u,v)$ in $G$ that is oriented from $u$
to $v$, 
\begin{equation}
f_{e}=\frac{\phi_{v}-\phi_{u}}{r_{e}}.\label{eq:potential_flow_def}
\end{equation}
In other words, a $\sigma$-flow $f$ is an electrical $\sigma$-flow
iff it is \emph{induced} via (\ref{eq:potential_flow_def}) by some
vertex potentials $\vphi$. (Note that orientation of edges matters
in this definition.)

Using vertex potentials, we are able to express the energy $\energy rf$
(see (\ref{eq:def_energy_flow})) of an electrical $\sigma$-flow
$f$ in terms of the potentials $\vphi$ inducing it as 
\begin{equation}
\energy rf=\sum_{e=(u,v)}\frac{(\phi_{v}-\phi_{u})^{2}}{r_{e}}.\label{eq:def_energy_potentials}
\end{equation}

\subsection{Laplacian Solvers}

A very important algorithmic property of electrical flows is that
one can compute very good approximations of them in nearly-linear
time. Below, we briefly describe the tools enabling that.

To this end, let us recall that electrical $\sigma$-flow is the (unique)
$\sigma$-flow induced by vertex potentials via (\ref{eq:potential_flow_def}).
So, finding such a flow boils down to computing the corresponding
vertex potentials $\vphi$. It turns out that computing these potentials
can be cast as a task of solving certain type of linear system called
\emph{Laplacian} systems. To see that, let us define the \emph{edge-vertex
incidence matrix} $B$ being an $n\times m$ matrix with rows indexed
by vertices and columns indexed by edges such that 
\[
B_{v,e}=\begin{cases}
1 & \text{if \ensuremath{e\in E^{+}(v)},}\\
-1 & \text{if \ensuremath{e\in E^{-}(v)},}\\
0 & \text{otherwise.}
\end{cases}
\]

Now, we can compactly express the flow conservation constraints (\ref{eq:conservation_constraints})
of a $\sigma$-flow $f$ (that we view as a vector in $\mathbb{R}^{m}$)
as 
\[
Bf=\sigma.
\]

On the other hand, if $\vphi$ are some vertex potentials, the corresponding
flow $f$ induced by $\vphi$ via (\ref{eq:potential_flow_def}) (with
respect to resistances $r$) can be written as 
\[
f=R^{-1}B^{\top}\vphi,
\]
where again $R$ is a diagonal $m\times m$ matrix with $R_{e,e}:=r_{e}$,
for each edge $e$.

Putting the two above equations together, we get that the vertex potentials
$\vphi$ that induce the electrical $\sigma$-flow determined by resistances
$r$ are given by a solution to the following linear system 
\begin{equation}
BR^{-1}B^{\top}\vphi=L\vphi=\sigma,\label{eq:elec_flow_lin_system}
\end{equation}
where $L:=BR^{-1}B^{T}$ is the (weighted) \emph{Laplacian $L$} of
$G$ (with respect to the resistances $r$). One can easily check
that $L$ is a symmetric $n\times n$ matrix indexed by vertices of
$G$ with entries given by 
\begin{equation}
L_{u,v}=\begin{cases}
\sum_{e\in E(v)}1/r_{e} & \text{if \ensuremath{u=v},}\\
-1/r_{e} & \text{if \ensuremath{e=(u,v)\in E}, and}\\
0 & \text{otherwise.}
\end{cases}\label{eq:def_of_laplacian}
\end{equation}

One can see that the Laplacian $L$ is not invertible, but -- as long
as, the underlying graph is connected -- its null-space is one-dimensional
and spanned by the all-ones vector. As we require our demand vectors
$\sigma$ to have its entries sum up to zero (otherwise, no $\sigma$-flow
can exist), this means that they are always orthogonal to that null-space.
Therefore, the linear system (\ref{eq:elec_flow_lin_system}) has
always a solution $\vphi$ and one of these solutions\footnote{Note that the linear system (\ref{eq:elec_flow_lin_system}) will
have many solutions, but any two of them are equivalent up to a translation.
So, as the formula (\ref{eq:potential_flow_def}) is translation-invariant,
each of these solutions will yield the same unique electrical $\sigma$-flow.} is given by 
\[
\vphi=L^{+}\sigma,
\]
where $L^{+}$ is the Moore-Penrose pseudo-inverse of $L$.

Now, from the algorithmic point of view, the crucial property of the
Laplacian $L$ is that it is symmetric and \emph{diagonally dominant},
i.e., for any $v\in V$, $\sum_{u\neq v}|L_{u,v}|\leq L_{v,v}$. This
enables us to use fast approximate solvers for symmetric and diagonally
dominant linear systems to compute an approximate electrical $\sigma$-flow.
Namely, there is a long line of work \cite{SpielmanT04,KoutisMP10,KoutisMP11,KelnerOSZ13,CohenKMPPRX14,KyngLPSS16,KyngS16}
that builds on an earlier work of Vaidya \cite{Vaidya90} and Spielman
and Teng \cite{SpielmanT03}, that designed an SDD linear system solver
that implies the following theorem. 
\begin{thm}
\label{thm:vanilla_SDD_solver} For any $\epsilon>0$, any graph $G$
with $n$ vertices and $m$ edges, any demand vector $\sigma$, and
any resistances $r$, one can compute in $\tilde{O}(m\log m\log\epsilon^{-1})$
time vertex potentials $\tvphi$ such that $\|\tvphi-\vphi^{*}\|_{L}\leq\epsilon\|\vphi^{*}\|_{L}$,
where $L$ is the Laplacian of $G$, $\vphi^{*}$ are potentials inducing
the electrical $\sigma$-flow determined by resistances $r$, and
$\|\vphi\|_{L}:=\sqrt{\vphi^{\top}L\vphi}$. 
\end{thm}
To understand the type of approximation offered by the above theorem,
observe that $\|\vphi\|_{L}^{2}=\vphi^{\top}L\vphi$ is just the energy
of the flow induced by vertex potentials $\vphi$. Therefore, $\|\tvphi-\vphi^{*}\|_{L}$
is the energy of the electrical flow $\bar{f}$ that ``corrects''
the vertex demands of the electrical $\tilde{\sigma}$-flow induced
by potentials $\tvphi$, to the ones that are dictated by $\sigma$.
So, in other words, the above theorem tells us that we can quickly
find an electrical $\tilde{\sigma}$-flow $\tilde{f}$ in $G$ such
that $\tilde{\sigma}$ is a slightly perturbed version of $\sigma$
and $\tilde{f}$ can be corrected to the electrical $\sigma$-flow
$f^{*}$ that we are seeking, by adding to it some electrical flow
$\bar{f}$ whose energy is at most $\epsilon$ fraction of the energy
of the flow $f^{*}$. (Note that electrical flows are linear, so we
indeed have that $f^{*}=\tilde{f}+\bar{f}$.) As we will see, this
kind of approximation is completely sufficient for our purposes.

\subsection{Bipartite b-Matchings}

For a given weighted bipartite graph $G=(V,E)$ with $V=P\cup Q$
-- where $P$ and $Q$ are the two sets of bipartition -- as well
as, a \emph{demand vector} $b\in\mathbb{R}_{+}^{V}$, a \emph{perfect
$b$-matching} is a vector $x\in\mathbb{R}_{+}^{V}$ such that $\sum_{e\in E(v)}x_{e}=b_{v}\text{ for all }v\in V.$
A perfect $b$-matching is a generalization of perfect matching; in
the particular case where all $b$'s equal $1$, integer $b$-matchings
(or $1$-matchings) are exactly perfect matchings. $x$ is called
\emph{$b$-matching} if the equality in the above equation is satisfied
with inequality, i.e., we have just $\sum_{e\in E(v)}x_{e}\le b_{v}\text{ for all }v\in V$.
For perfect $b$-matchings we usually require that $b(P)=b(Q)$ as
otherwise they trivially do not exist.

A \emph{weighted perfect bipartite $b$-matching problem} is a problem
in which, given a weighted bipartite graph $G=(V,E,w)$ our goal is
to either return a perfect $b$-matching in $G$ that has minimum
weight, or conclude that there is no perfect $b$-matching in $G$.
We say that a $b$-matching $x$ in a graph $G$ is \emph{$h$-near}
if the size of $x$ is within an additive factor of $h$ of the size
of a perfect \emph{$b$}-matching. The dual problem to the weighted
perfect bipartite $b$-matching is a \emph{$b$-vertex packing problem}
where we want to find a vector $y\in\mathbb{R}^{V}$ satisfying the
following LP 
\begin{eqnarray*}
\max &  & \sum_{v\in V}y_{v}b_{v},\\
 &  & y_{u}+y_{v}\le w_{uv}\quad\forall uv\in E.
\end{eqnarray*}

Also, we define the \emph{size} of a $b$-matching $x$ to be $\norm x_{1}/2$.
A \emph{$h$-near $b$-matching }is a \emph{b}-matching with size
at least $\norm x_{1}/2-h$. Finally, we observe that a bipartite$b$-matching
instance can be reinterpreted as a unit capacitated $\sigma$-flow
instance just by setting 

\[
\sigma_{v}=\begin{cases}
b_{v} & \textnormal{if }\mathrm{\mathrm{\mathit{v\in P}}},\\
-b_{v} & \textnormal{otherwise.}
\end{cases}
\]

and leaving costs unchanged. We require this alternative view, since
our interior-point algorithm will work with $\sigma$-flows, but the
rounding that will need to be performed at the end is done in the
$b$-matching view.

\section{\label{sec:reduction}Reducing Minimum-Cost Flow to Bipartite $b$-matching}

In this section we show how to convert an instance of unit capacity
min-cost flow into an instance of min-cost $b$-matching. This is
done via a simple combinatorial reduction similar to the one in \cite{Madry13}.
As noted in the preliminaries, bipartite $b$-matchings can be reinterpreted
as $\sigma$-flows. This alternative view will be useful for us, as
our interior point method will work with $\sigma$-flows, whereas
it is easier to repair a near-optimal solution in the $b$-matching
view.

We first show a straightforward reduction from min-cost flow to $b$-matching.
One desirable feature of this reduction is that the obtained $b$-matching
instance does not contain upper capacities on arcs, since these are
going to be implicitly encoded by properly setting the demands. The
part of lemma that refers to half-integral flow and half-integral
matching will be essential for the initialization step. 
\begin{lem}
Given a directed graph $G=(V,E,c)$ and a demand vector $\sigma$,
one can construct in linear time a bipartite graph $G'=(V',E',c')$,
$V'=P\cup Q$ along with a demand vector $b'$ such that given a minimum
cost $b'$-matching in $G'$, one can reconstruct a flow $f$ that
routes demand $\sigma$ in $G$ with minimum cost. Moreover, if flow
$f=\frac{1}{2}\cdot\vec{1}$ is feasible in $G$ then $x=\frac{1}{2}\cdot\vec{1}$
is a feasible fractional $b$-matching in $G'$. \label{lem:reduction-to-bmatching}\end{lem}
\begin{proof}
Let $P=V$ and $Q=E$. For each arc $(u,v)\in E$, let $e_{uv}$ be
the corresponding ``edge'' vertex in $Q$. Create arcs $(u,e_{uv})$
with cost $c_{uv}$, and $(v,e_{uv})$ with cost $0$. For each $e_{uv}\in Q$
set demand $b'(e_{uv})=1$. For each $v\in P$, set demand $b'(v)=\sigma(v)+\textnormal{deg}_{in}^{G}(v)$.
Let us now argue that the solution to $b'$-matching encodes a valid
flow in $G.$ Observe that in the $b'$-matching instance each vertex
$e_{uv}$ can be in two states: it is either matched to $u$ or to
$v$. When $e_{uv}$ is matched to $u$ we ``read'' that there is
one unit of flow on arc $(u,v)$, whereas when $e_{uv}$ is matched
to $v$ we ``read'' that there is no flow on arc $(u,v).$ With
this interpretation flow conservation constraints (\ref{eq:conservation_constraints})
are satisfied. 

Now assume that the flow $f=\frac{1}{2}\cdot\vec{1}$ is feasible
in $G$ and consider the $b$-matching $x=\frac{1}{2}\cdot\vec{1}$.
We observe that feasibility of $f$ implies that $x$ is feasible
for each vertex in $P$, as each such vertex has the same number of
incident edges and the same demand as the corresponding vertex in
$G$. On the other hand, by construction vertices in $Q$ have demand
$-1$ and two incoming edges, what settles feasibility of $x$ for
them. 
\end{proof}
The effect of this reduction on a single arc is presented in the figure
below. Note that vertices in $P$ correspond to \textit{vertices}
from the original graph, whereas vertices in $Q$ correspond to \textit{arcs}
form the original graph. Essentially, every arc $(u,v)$ in $G$ adds
a demand pair of $(1,1)$ on the pair of vertices $(v,e_{uv})$ in
the $b$-matching instance. The amount of flow routed on $(v,e_{uv})$
corresponds to the residual capacity of the arc $(u,v)$ in the original
graph. 
\begin{figure}[H]
\centering{}\includegraphics[bb=0bp 250bp 1024bp 768bp,clip,scale=0.4]{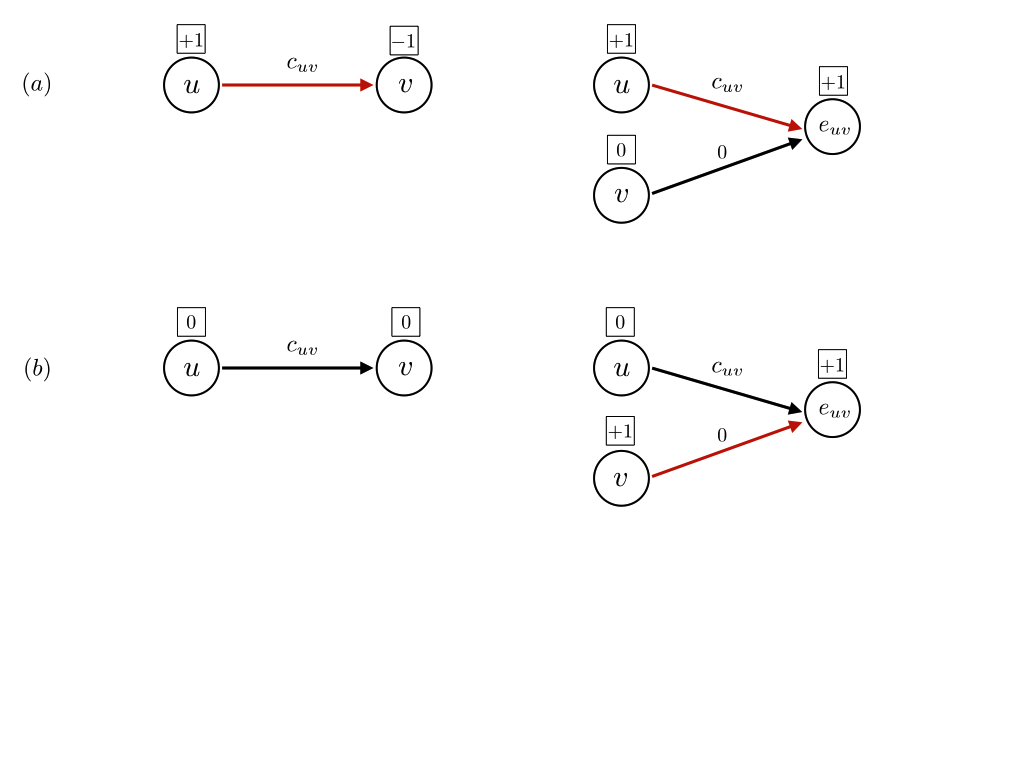}\caption{The left side contains two $\sigma$-flow instances, the right side
contains their corresponding reductions to $b$-matching. In (a) we
can see that the unit flow from $u$ to $v$ gets routed from $u$
to $e_{uv}$ in the reduction. In (b), since there is no flow from
$u$ to $v$, the new demand added by the reduction gets routed directly
from $v$ to $e_{uv}$.}
\end{figure}

While this reduction rephrases the problem into a form that is amenable
to our interior-point framework, the remaining issue is that we have
to be able to start the algorithm with a feasible solution where all
the flows on arcs are similar in magnitude (this enables us to enforce
the centrality property defined in \ref{par:Centrality}). Ideally,
we should be able to obtain a feasible instance simply by placing
half a unit of flow on every arc of the $b$-matching instance. While
doing so clearly satisfies the demands of vertices in $Q$, demand
violations might occur on vertices in $P$.

We can easily fix this problem by adding one extra vertex in the original
graph, along with a set of arcs with costs chosen sufficiently high,
that the optimal solution will never consider them. The goal is to
add the extra arcs such that flowing $\frac{1}{2}$ on every arc satisfies
the demand. This instance can then be converted to a $b$-matching
instance with the same property. The reduction is described by the
following lemma:
\begin{lem}
Given a directed graph $G=(V,E,c)$ and an integral demand vector
$\sigma$, one can construct in linear time a graph $G'=(V',E',c')$
along with an integral demand vector $\sigma'$ such that the demand
is satisfied by placing $\frac{1}{2}$ units of flow on every arc.
Furthermore, given a solution to the min-cost $\sigma'$-flow satisfying
demand $\sigma'$ in $G'$, one can reconstruct in linear time a solution
$f$ that routes demand $\sigma$ in $G$ with minimum cost.\label{lem:reduction-to-balance}\end{lem}
\begin{proof}
Create one extra vertex $v_{\textnormal{aux}}$ with demand $0$.
Then, for all $v\in V$, let $t(v)=\sigma(v)+\frac{1}{2}\cdot\deginG G(v)/2-\frac{1}{2}\cdot\degoutG G(v)$
be the residual demand corresponding to the flow that has value $\frac{1}{2}$
everywhere. 

Fix this residual by creating $\abs{2t(v)}$ parallel arcs $(v_{\textnormal{aux}},v)$
with costs $\norm c_{1}$, for each vertex with $t(v)<0$, respectively
$\abs{2t(v)}$ parallel arcs $(v,v_{\textnormal{aux}})$ with costs
$\norm c_{1}$ for each vertex $v$ with $t(v)>0$. This enforces
our condition to be satisfied for all vertices in $V$.

Also note that $v_{\textnormal{aux}}$ has an equal number of arcs
entering and leaving it, since the sum of residuals at vertices in
$V$ equals the sum of degree imbalances, plus the sum of demands,
both of which are equal to $0$. More precisely, $\deg_{in}(v_{\textnormal{aux }})=\sum_{v:t(v)>0}2t(v)$,
and $\deg_{out}(v_{\textnormal{aux }})=\sum_{v:t(v)<0}-2t(v)$; since
$0=\sum_{v}t(v)=\sum_{v:t(v)>0}t(v)-\sum_{v:t(v)<0}\left(-t(v)\right)=\deginG{G'}(v_{\textnormal{aux}})/2-\degoutG{G'}(v_{\textnormal{aux}})/2$,
the vertex $v_{\textnormal{aux}}$ is balanced; hence $t(v_{\textnormal{aux}})=0$,
and the residual demand is $0$ at all vertices in the graph.

Finally, observe that a flow in $G$ satisfying $\sigma$ is a valid
flow in $G'$ for $\sigma'$ that does not use any edge incident to
$v_{\textnormal{aux}}$. This flow has cost smaller than $\norm c_{1}$.
Hence, a min-cost $\sigma'$-flow in $G'$, if it has cost smaller
than $\norm c_{1}$, gives a min-cost $\sigma$-flow in $G$.
\end{proof}
A pictorial description of the reduction is presented in the figure
below.
\begin{figure}[H]
\centering{}\includegraphics[bb=0bp 370bp 1024bp 740bp,clip,scale=0.4]{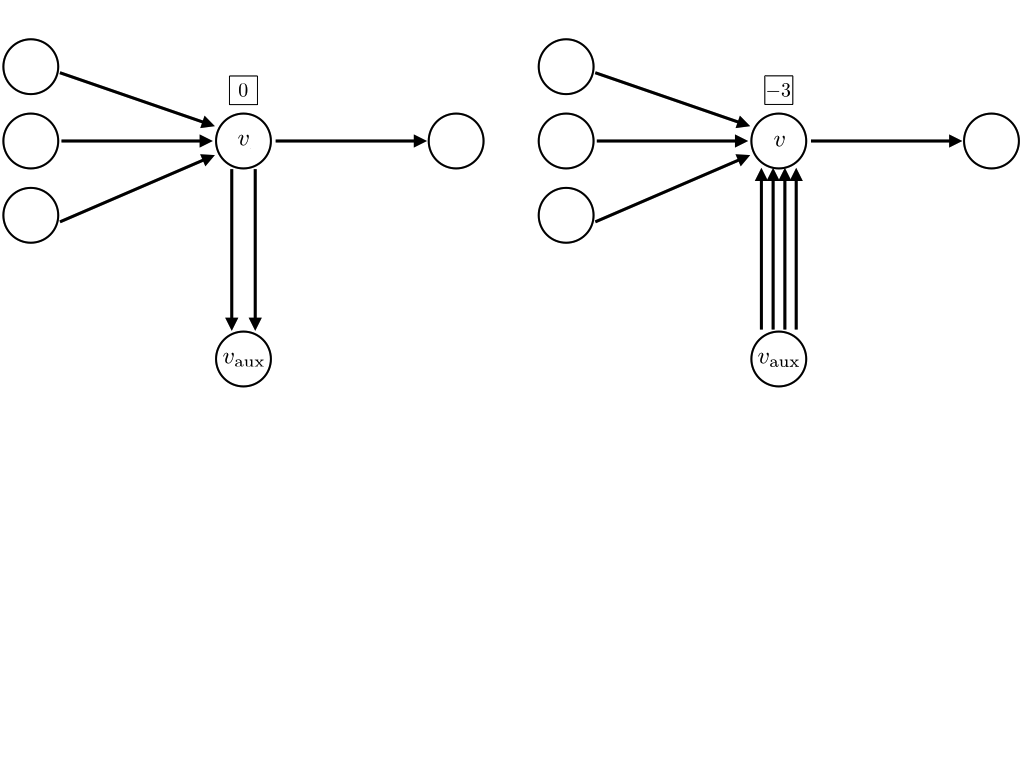}\caption{Two examples of balancing a vertex using extra arcs connected to $v_{\textnormal{aux}}$.
In the first case, adding two arcs to $v_{\textnormal{aux}}$ makes
the residual at $v$ equal to $0$, when routing $1/2$ on every arc,
since the in-degree of $v$ is equal to its out-degree, and there
is no demand on $v$. In the second case, we add four arcs from $v_{\textnormal{aux}}$
to $v$ in order for the net flow at $v$ (when routing $1/2$ on
every arc) to match the demand $-3$.}
\end{figure}

We can use these two reductions together, first making the flow $f=\frac{1}{2}\cdot\vec{1}$
feasible, then converting the instance to a $b$-matching instance
that can be directly changed back to $\sigma$-flow instance. As a
result we obtain an instance where constructing a feasible starting
solution that is required for our interior point method is straightforward.

\section{Our Interior-Point Method Framework}

\label{sec:basic_framework}

In this section we describe our interior point method framework for
solving our instance of the uncapacitated minimum-cost $\sigma$-flow
problem that results from casting the bipartite $b-$matching problem
instance we produced in Section \ref{sec:reduction} as an instance
of the minimum-cost $\sigma-$flow problem. Our basic setup is largely
following the framework used by M\k{a}dry \cite{Madry13}. The crucial
differences emerge only later on, in Sections \ref{sec:longer_steps}
and \ref{sec:preconditioning}, when we refine it to obtain our desired
running time improvements.

\subsection{Primal and Dual Solutions and Centrality}

Our framework will be inherently primal dual. That is, we will always
maintain a feasible primal dual solution $(f,y)$. Here, the primal
solution $f$ is simply a $\sigma$-flow, i.e., a flow with demands
$\sigma$ and its feasibility condition is that $f_{e}\geq0$ for
each arc $e$, i.e., the flow never flows against the orientation
of the arc. The dual solution $y$, on the other hand, corresponds
to embedding of the vertices of our graph into the line, i.e., $y$
assigns a real value $y_{v}$ to each vertex $v$. The dual feasibility
condition is that, for each arc $e=(u,v)$, we have that its \emph{slack}
$s_{e}:=c_{e}+y_{u}-y_{v}$ is non-negative, i.e., $s_{e}\geq0$,
for each arc $e$. Intuitively, this means that no arc $e$ is stretched
in the embedding defined by $y$ more than its actual cost $c_{e}$.
(Note that here we measure stretch only in the direction of the arc
orientation.)

Observe that the dual solution $y$ is define uniquely (up to a translation)
by the slack vector $s$ and the arc costs $c$. So, we will sometime
find it more convenient to represent our dual solution in terms of
the slack vector $s$ instead of the embedding vector $y$. From this
perspective, the dual feasibility condition is simply non-negativity
of all the slacks $s$. (In fact, we will use $y$ and $s$ representation
interchangeably, depending on which one of them is more convenient
in given context.)

\paragraph{Duality Gap.}

A very convenient feature of working with primal dual solutions is
that they provide a very natural way of measuring optimality: the
duality gap, i.e., the difference between the upper bound on the cost
of the optimal solution provided by current primal feasible solution
$f$ and the lower bound on the optimal cost provided by our current
dual solution $s$. It turns out that the duality gap of a primal
dual solution $(f,s)$ is exactly

\[
f^{\top}s=\sum_{e}f_{e}s_{e}=\sum_{e}\mu_{e},
\]

where $\mu_{e}:=f_{e}s_{e}$ is the contribution of the arc $e$ to
that duality gap. Observe that by our primal dual feasibility condition
we have always that $\mu_{e}\geq0$, for each arc $e$.

\paragraph{Centrality. \label{par:Centrality}}

In the light of the above, the duality gap $\norm{\mu}_{1}=\sum_{e}\mu_{e}$
provides a natural measure of progress for our algorithm. In particular,
our goal is simply to find a primal dual feasible solution $(f,s)$
with $\norm{\mu}_{1}$ sufficiently close to $0$, and, from this
perspective, any way of updating that primal dual solution that leads
to reduction of the duality gap should be desirable. 

However, for reasons that will become clear later, we will insist
on reducing the duality gap $\norm{\mu}_{1}$ in a more restricted
manner. To define this, let us first associate with each arc $e$
a value $\nu_{e}$ that we will refer to as the \emph{measure} of
\emph{$e$}. We will always make sure that $\nu_{e}\geq1$, for each
arc $e$, and that the total sum of all the measures is not too large.
Specifically, we want to maintain the following invariant that ensures
that the average measure of an arc is at most $3$.

\vspace{10bp}

\textbf{Invariant 1.} We always have that $\norm{\nu}_{1}=\sum_{e}\nu_{e}\leq3m$.

\vspace{10bp}

As it turns out, once we introduced the notion of arc measure, it
will be much more convenient to introduce analogues of the standard
$\ell_{p}$-norms that are reweighted by the measures of arcs. To
this end, for a given vector $x\in\mathbb{R^{\text{m}}}$, and measure
vector $\nu$, let us define

\begin{equation}
\|x\|_{\nu,p}:=\left(\sum_{e}\nu_{e}\abs{x_{e}}^{p}\right)^{\frac{1}{p}}\label{eq:def_nu_l_p_norm}
\end{equation}

We will sometimes extend this notation to refer to the weighted norm
of a subset of indices in the support. More specifically, given $S\subseteq\left\{ 1,\dots,m\right\} $,
we will call

\begin{equation}
\|x_{S}\|_{\nu,p}:=\left(\sum_{e\in S}\nu_{e}\abs{x_{e}}^{p}\right)^{\frac{1}{p}}\label{eq:def_nu_l_p_norm-1}
\end{equation}

Also, we will extend our notation and use $(f,s,\nu)$ to denote a
primal dual solution $(f,s)$ with its corresponding measure vector
$\nu$.

Now, we can make precise the additional constraint on our primal dual
solutions $(f,s,\nu)$ that we want to maintain. Namely, we want each
solution \textbf{$(f,s,\nu)$} to be $\widehat{\mu}-centered$ (or,
simply, \emph{centered}), i.e., we want that, for each arc $e$, 

\begin{equation}
\mu_{e}=\nu_{e}\widehat{\mu},\label{eq:def_centrality}
\end{equation}

where $\widehat{\mu}$ is a normalizing value we will refer to as
\emph{average duality gap}. Intuitively, centrality means that each
arc's contribution to the duality gap is exactly proportional to its
measure.\footnote{The framework in \cite{Madry13} works with a slightly relaxed notion
of centrality. However, we deviate from that here.}

Note that the above notions enable us to express the duality gap of
a solution $(f,y,\nu)$ as exactly $\sum_{e}\nu_{e}\widehat{\mu}$,
which by Invariant $1$ is at most $3m\widehat{\mu}$. That is, we
have that

\begin{equation}
\sum_{e}\mu_{e}=\sum_{e}\nu_{e}\widehat{\mu}\leq3m\widehat{\mu}\label{eq:average_duality_gap_bound}
\end{equation}

Consequently, we can view $\widehat{\mu}$ as a measure of our progress
- driving it to be smaller translates directly into making the duality
gap smaller too.

\subsection{Making Progress with Electrical Flow Computations\label{sub:Making-Progress}}

Once we setup basic definitions, we are ready to describe how we initialize
our framework and then how we can use electrical flow computations
to gradually improve the quality, i.e., the average duality gap $\widehat{\mu}$
of our $\widehat{\mu}$-centered primal dual solution ($f,y,\nu)$.

\paragraph{Initialization. }

As we want our solutions ($f,y,\nu)$ to be always centered, initialization
of our framework, i.e., finding the initial centered primal dual solution,
might be difficult. Fortunately, one of the important properties of
the reduction we performed in Section \ref{sec:reduction} is that
a centered primal dual feasible solution of the resulting uncapacitated
minimum-cost $\sigma$-flow instance can be specified explicitly. 
\begin{lem}
\label{lem:initial-centering}Given a bipartite graph $G=(V,E,c)$
along with an integral demand vector $\sigma$ and a subset of vertices
$P$ such that that for all $v\in P$ we have $\sigma_{v}=\deg(v)/2$,
whereas for all $v\not\in P$ we have $\deg(v)=2$, one can construct
in linear time a feasible primal-dual set of variables $(f,s)$ that
satisfy the centrality bound for $\widehat{\mu}=\norm c_{\infty}$
and $\norm{\nu}_{1}\le m$.\end{lem}
\begin{proof}
Since $\sigma_{v}=\deg(v)/2$ for all $v\in P$, while all $w\not\not\in P$
have degree precisely $2$, we can set $f=\frac{1}{2}\cdot\vec{1}$
and have all the demands satisfied exactly. Moreover, we set the dual
variables $y_{v}=\norm c_{\infty}$ for all $v\in P$, and $y_{w}=0$
for all $w\not\in P$. This way the slacks $s_{vw}=c_{vw}-y_{w}+y_{v}$
are all within the range $[\norm c_{\infty},2\norm c_{\infty}]$.
We set $\nu_{e}=\frac{s_{e}}{2\norm c_{\infty}}$ and $\widehat{\mu}=\norm c_{\infty}$
so that 

\[
\mu_{e}=f_{e}s_{e}=\frac{1}{2}s_{e}=\frac{s_{e}}{2\norm c_{\infty}}\norm c_{\infty}=\nu_{e}\widehat{\mu},
\]

and

\[
\norm{\nu}_{1}=\sum_{e}\nu_{e}=\sum_{e}\frac{s_{e}}{2\norm c_{\infty}}\le\sum_{e}1=m.
\]

\end{proof}

\paragraph{Taking an Improvement Step.\label{par:Taking-an-Improvement}}

Let us fix some $\widehat{\mu}$-centered primal dual solution ($f,y,\nu)$
and let us define resistances $r$ to be equal to 
\begin{equation}
r_{e}:=\frac{1}{\hat{\mu}}\cdot\frac{s_{e}}{f_{e}}=\frac{1}{\hat{\mu}}\cdot\frac{\mu_{e}}{f_{e}^{2}}=\frac{\nu_{e}}{f_{e}^{2}},\label{eq:def_r_e}
\end{equation}

for each arc $e$. (Note that $f$ has to always be positive due to
centrality condition, and thus these resistances are well-defined.)

The fundamental object that will drive our improvements of the quality
of our current primal dual solution will be the electrical $\sigma-$flow
$\elf$ determined by the above resistances \textbf{$r$}. For future
reference, we will call the electrical flow $\elf$ \emph{associated
with} ($f,s,\nu)$. The key property of that electrical flow is that
it will enable us to update our primal and dual solutions \emph{simultaneously}.
That is, we can use the flow itself to update the primal solution
$f$, and we can use the vertex potentials $\pot$ that induced $\elf$
to update our dual solution $s$. Specifically, our main improvement
update step, for each arc $e=(u,v)$ is:

\begin{eqnarray*}
f_{e}' & := & (1-\delta)f_{e}+\delta\hat{f}_{e},\\
s_{e}' & := & s_{e}-\frac{\delta}{(1-\delta)}\left(\widehat{\phi}_{v}-\widehat{\phi}_{u}\right),
\end{eqnarray*}

where $\delta$ is a \emph{step size} that we will choose later.
\begin{rem}
The step derived from the standard primal-dual interior-point method
computes an electrical flow along with potentials determined by resistances
$\frac{s_{e}}{f_{e}}$, which are off by precisely a factor of $\hat{\mu}$
from the resistances we consider in this paper. However, scaling all
resistances by the same factor has no effect on the electrical flow
or the potentials produced. Setting resistances the way we do in (\ref{eq:def_r_e})
has the benefit that it will enable us to relate the electrical energy
with another quantity of interest without having to carry along the
extra $\frac{1}{\hat{\mu}}$ factor, as we will see in Lemma \ref{lem:linking_rho_to_energy}.
\end{rem}
Intuitively, this update step mixes the electrical flow $\elf$ with
the current solution $f$ by taking a convex combination of them.
(Note that the resulting flow is guaranteed to be a $\sigma$- flow
in this way.) On the other hand, the dual update corresponds to updating
the line embedding of each vertex by adding an appropriately scaled
vertex potential to it. 

It is worth pointing out that the electrical flow $\elf$ is inherently
undirected. So, it is not a priori clear if the flow $f'$ resulting
from the above update is even feasible. As a result, we will need
to ensure, in particular, that the step size $\delta$ is chosen to
be small enough so as to ensure that $f'$ is still feasible. (In
fact, as we will see shortly, there are some even stronger restrictions
on the value of $\delta$. So, the feasibility will be enforced implicitly.)

\paragraph*{Congestion Vector.\label{par:Congestion-Vector.}}

A notion that will be extremely useful in analyzing our improvement
step and the performance of our algorithm in general is the notion
of congestion vectors. Specifically, given the electrical $\sigma$-flow
$\widehat{f}$ associated with our solution ($f,y,\nu)$, let us define
\emph{congestion $\rho_{e}$ }of an arc $e$ as 

\begin{equation}
\rho_{e}:=\frac{\abs{\elf_{e}}}{f_{e}}\label{eq:def_congestion_vector}
\end{equation}

Now, observe that we can express the duality contribution $\mu_{e}'$
of an arc $e$ in the new solution $(f',s')$ as

\begin{align*}
\mu_{e}' & =f_{e}'s_{e}'=\left((1-\delta)f_{e}+\delta\hat{f}_{e}\right)\left(s_{e}-\frac{\delta}{(1-\delta)}\left(\widehat{\phi}_{v}-\widehat{\phi}_{u}\right)\right)\\
 & =(1-\delta)f_{e}s_{e}-\delta f_{e}\left(\widehat{\phi}_{v}-\widehat{\phi}_{u}\right)+\delta\hat{f}_{e}s_{e}-\frac{\delta^{2}}{1-\delta}\hat{f}_{e}\left(\widehat{\phi}_{v}-\widehat{\phi}_{u}\right)\\
 & =(1-\delta)f_{e}s_{e}-\delta f_{e}\cdot\hat{f}_{e}\frac{s_{e}}{f_{e}}+\delta\hat{f}_{e}s_{e}-\frac{\delta^{2}}{1-\delta}\hat{f}_{e}\cdot\hat{f}_{e}\frac{s_{e}}{f_{e}}\\
 & =(1-\delta)\mu_{e}-\frac{\delta^{2}}{1-\delta}\mu_{e}\rho_{e}^{2}
\end{align*}

So, if we ignore the second-order term in $\delta$, the duality gap
contribution of each arc $e$ goes down at the same rate. In this
way, not only the duality gap gets reduced by a factor of $(1-\delta)$
but also the centrality of the solution would be perfectly preserved.

However, we cannot really ignore the second-order term and this term
will make our solution lose centrality. Fortunately, one can show
that as long as the total degradation of centrality condition is not
too large one can easily correct it with a small number of electrical
flow computations. Specifically, for the correction to be possible,
we need to have that the total $\ell_{2}^{2}-$norm of the degradations
(measured with respect to measure $\nu$ and normalized by the duality
gap contributions $\mu$) has to be a small constant. That is, we
need that

\begin{equation}
\sum_{e}\nu_{e}\left(\frac{\mu_{e}'}{(1-\delta)\mu_{e}}-1\right)^{2}=\sum_{e}\nu_{e}\left(\frac{\delta^{2}\mu_{e}\rho_{e}^{2}}{(1-\delta)^{2}\mu_{e}}\right)^{2}\frac{\delta^{4}}{(1-\delta)^{4}}\sum_{e}\nu_{e}\rho_{e}^{4}=\frac{\delta^{4}}{(1-\delta)^{4}}\|\rho\|_{\nu,4}^{4}\leq\frac{1}{256},\label{eq:centrality}
\end{equation}

which implies that it is sufficient to have

\[
\delta\leq\frac{1}{8\cdot\|\rho\|_{\nu,4}},
\]
i.e., that the step size $\delta$ should be bounded by the $\ell_{4}$
norm of the congestion vector $\rho$. The following theorem makes
these requirements, as well as the result of the full improvement
step, precise. Its complete proof can be found in Appendix \ref{sec:proof-progress-step}.
\begin{thm}
\label{thm:l_4_improvement_step}Let $(f,s,\nu)$ be a $\widehat{\mu}-$centered
solution and let \textbf{$\rho$} the congestion vector of the electrical
flow $\widehat{f}$ associated with that solution. For any $\delta>0$
such that 

\[
\delta\leq\min\left\{ \frac{1}{8\cdot\|\rho\|_{\nu,4}},\frac{1}{8}\right\} ,
\]

we can compute in $\tilde{O}(m)$ time a $\widehat{\mu}'-$centered
solution ($f',s',\nu')$, such that $\nu'=\nu$, $\widehat{\mu'}\leq(1-\delta)\widehat{\mu}$,
and, for each arc $e,$

\[
r_{e}'=\frac{1}{\hat{\mu}}\cdot\frac{s_{e}'}{f_{e}'}\geq\left(1+4\cdot\delta\rho_{e}+\kappa_{e}\right)^{-1}r_{e},
\]

where $\kappa$ is a vector with $\|\kappa\|_{\nu,2}\leq1$.
\end{thm}

\subsection{A Simple $O\left(\sqrt{m}\log W\right)$-iteration Bound}

Once the $\ell_{4}$ norm bound provided in Theorem \ref{thm:l_4_improvement_step}
is established we are already able to prove in a simple way that our
algorithm needs at most $O\left(\sqrt{m}\log W\right)$ iterations
to compute the optimal solution, making its total running time be
at most $O\left(m^{3/2}\log W\right)$. To achieve that, we just need
to argue that we always have that

\begin{equation}
\|\rho\|_{\nu,4}\leq O\left(\sqrt{m}\right).\label{eq:sqrt_bound_on_rho}
\end{equation}
Once this is established, by Theorem \ref{thm:l_4_improvement_step},
we know that we can always take $\delta=\Omega(m^{-1/2})$ and thus
make $\widehat{\mu}$ decrease by a factor of $(1-\delta$) in each
iteration. So, after $O\left(\sqrt{m}\log W\right)$ iterations, $\widehat{\mu}$
and thus the duality gap becomes small enough that a simple rounding
(see Section \ref{sec:repair}) will recover the optimal solution. 

Now, to argue that \ref{eq:sqrt_bound_on_rho} indeed holds, we first
notice that we can always upper bound $\ell_{4}$ norm $\|\rho\|_{\nu,4}$
by the $\ell_{2}$ norm $\|\rho\|_{\nu,2}$ and then bound the latter
norm instead. Next, as it turns out, we can tie the energy $\mathcal{E}(\elf)$
of the electrical flow $\elf$ associated with a given solution $(f,s,\nu)$
to the corresponding $\ell_{2}$ norm $\|\rho\|_{\nu,2}$. Specifically,
we have the following lemma. 
\begin{lem}
\label{lem:linking_rho_to_energy}For any centered solution $(f,s,\nu)$,
we have that

\[
\norm{\rho}_{\nu,2}^{2}=\mathcal{E}(\elf),
\]
where \textup{$\elf$ }\textup{\emph{is the electrical flow associated
with that solution and $\rho$ is its congestion vector.}}\end{lem}
\begin{proof}
Observe that by definition \ref{eq:def_congestion_vector} and \ref{eq:def_r_e},
we have that

\[
\norm{\rho}_{\nu,2}^{2}=\sum_{e}\nu_{e}\rho_{e}^{2}=\sum_{e}\nu_{e}\left(\frac{\widehat{f_{e}}}{f_{e}}\right)^{2}=\sum_{e}r_{e}\elf_{e}^{2}=\mathcal{E}(\elf)
\]

Note that we used (\ref{eq:def_r_e}) to write $r_{e}=\frac{\nu_{e}}{f_{e}^{2}}$,
which assumes centrality.
\end{proof}
Due to this identity, we can view the $\ell_{2}$ norm $\norm{\rho}_{\nu,2}^{2}$
as energy. Finally, we can use the bound from Invariant 1 to show
that the energy $\mathcal{E}(\elf)$ and thus the norm $\norm{\rho}_{\nu,2}$
can be appropriately bounded as well. 
\begin{lem}
\label{lem:rho_l_2_norm_bound}For a centered solution,$\norm{\rho}_{\nu,2}^{2}\leq\sum_{e}\nu_{e}=\norm{\nu}_{1}$.\end{lem}
\begin{proof}
By Lemma \ref{lem:linking_rho_to_energy} and (\ref{eq:def_r_e}),
we have that 
\begin{align*}
\norm{\rho}_{\nu,2}^{2} & =\mathcal{E}(\hat{f})\leq\mathcal{E}(f)=\sum_{e}r_{e}f_{e}^{2}=\sum_{e}\nu_{e}\left(\frac{f_{e}}{f_{e}}\right)^{2}=\sum_{e}\nu_{e}=\norm{\nu}_{1},
\end{align*}

where the inequality follows as $f$ is a $\sigma$-flow and, by the
virtue of being an electrical $\sigma$-flow, $\widehat{f}$ has to
have minimum among all the $\sigma$-flows. 
\end{proof}
Now, since by Invariant 1 , $\norm{\nu}_{1}\leq3m$, we can conclude
that 

\[
\|\rho\|_{\nu,4}^{2}\leq\|\rho\|_{\nu,2}^{2}\leq\norm{\nu}_{1}\leq3m
\]

and the bound \ref{eq:sqrt_bound_on_rho} follows.

With this upper bound on the $\norm{\rho}_{\nu,4}$ we can immediately
derive a bound on the running time required to obtain an exact solution.
This is summarized in the following theorem.
\begin{thm}
\label{thm:min-cost-flow-m32}We can produce an exact solution to
the unit-capacitated minimum cost $\sigma$-flow problem in $\tilde{O}\left(m^{3/2}\log W\right)$
time.\end{thm}
\begin{proof}
Given an instance of the unit-capacitated minimum cost $\sigma$-flow,
we can apply the reduction from Section \ref{sec:reduction} in linear
time. Then, Lemma \ref{lem:initial-centering} establishes the initial
centering with $\hat{\mu}=\norm c_{\infty}\leq W$. We previously
saw that $\norm{\rho}_{\nu,4}=O(\sqrt{m})$. Therefore, according
to Theorem \ref{thm:l_4_improvement_step}, we can set $\delta=1/O(\sqrt{m})$,
and reduce $\hat{\mu}$ by a factor of $1-\frac{1}{O(\sqrt{m})}$
with every interior-point iteration. Therefore, in $O\left(m^{1/2}\left(\log m+\log W\right)\right)$
iterations we reduce $\hat{\mu}$ to $O(m^{-3})$. Using the fact
that $\norm{\nu}_{1}\leq3m$, the duality gap of this solution will
be $\sum_{e}\nu_{e}\hat{\mu}\leq m^{-2}$. Note that each iteration
requires $\tilde{O}(1)$ electrical flow computations, and each of
them can be implemented in near-linear time, according to Theorem
\ref{thm:vanilla_SDD_solver}. 

Therefore in $\tilde{O}\left(m^{3/2}\log W\right)$ time, we obtain
a feasible primal-dual solution with duality gap less than $m^{-2}$.
This can easily be converted to an integral solution in nearly-linear
time using the method described in Section \ref{sec:repair}. Hence
the total running time of the algorithm is $\tilde{O}(m^{3/2}\log W)$.\end{proof}

\section{Taking Longer Steps}

\label{sec:longer_steps}

In the previous section, we established that the amount of progress
we can make in one iteration of our framework is limited by the $\ell_{4}$
norm of the congestion, $\norm{\rho}_{\nu,4}$ - see Theorem \ref{thm:l_4_improvement_step}.
We then showed (cf. (\ref{eq:sqrt_bound_on_rho})) that this norm
is always upper bounded by $O(\sqrt{m})$ , which resulted in our
overall $\widetilde{O}\left(m^{3/2}\log W\right)$ time bound. 

Unfortunately, a priori, this upper bound is tight, i.e., it indeed
can happen that $\norm{\rho}_{\nu,4}$=$\norm{\rho}_{\nu,2}$$=\Omega(\sqrt{m}).$
In fact, this is exactly the reason why all classic analyses of interior-point
algorithms are able to obtain only an $O(\sqrt{m})$ iteration bound. 

To circumvent this problem, \cite{Madry13} introduced a perturbation
technique into the framework. These perturbations target arcs that
contribute a significant fraction of the norm $\norm{\rho}_{\nu,4}$
, by increasing their resistance (and thus the corresponding energy),
in order to discurage the emergence of such high contributing arcs
in the future. A careful analysis done in \cite{Madry13} shows that
such perturbations indeed ensure that there are sufficiently many
iterations with relatively small $\norm{\rho}_{\nu,4}$ norm to guarantee
ability to take longer steps, and thus converge faster. Unfortunately
the changes to the underlying optimization problem that these perturbations
introduced, although sufficiently mild to enable obtaining a result
for unit-capacity maximum flow, were too severe to enable solving
any of the weighted variants that we aim to tackle here. 

Our approach will also follow the same general outline. The crucial
difference though is that we use a different perturbation technique,
along with a somewhat simpler analysis. This technique still achieves
the desired goal of increasing the resistance of the perturbed arcs.
However, in big contrast to the technique used by \cite{Madry13},
our technique does not affect the costs of those arcs -- it affects
only their measure. Also, as an added benefit, our perturbation treatment
simplifies the analysis significantly. 

We describe our preconditioning technique in Section \ref{sub:Perturbing-the-problem}.
Also, in the table below we present a general outline of our algorithm.
(This algorithm will be later modified further to include a preconditioning
step.) Observe that this algorithm uses a stronger, $\ell_{3}$ norm
criterion for whether to make a perturbation or a progress step, instead
of the $\ell_{4}$ norm criterion that Theorem \ref{thm:l_4_improvement_step}
suggests. As we will see, the reason for that is that maintaining
such an $\ell_{3}$ norm condition will enable us to have a sufficiently
tight control over the change after each progress step of our potential
function: the energy of electrical flows $\elf$ associated with our
primal dual solutions.

Our goal is to obtain an $\tilde{O}(m^{1/2-\eta})$ bound on the overall
number of iterations, where we fix $\eta$ to be

\[
\eta=\frac{1}{14}
\]

\begin{algorithm}[H]
\begin{enumerate}
\item initialize primal and dual variables $(f,y)$ (as shown in Section
\ref{sec:basic_framework});
\item repeat $\cT\cdot m^{1/2-\eta}$ times
\item $\quad$while $\Vert\rho\Vert_{\nu,3}>\crho\cdot m^{1/2-\eta}$
\item $\quad\quad$perform perturbation (as shown in Section \ref{sub:Perturbing-the-problem})
\item $\quad$perform progress steps (as shown in Section \ref{sub:Making-Progress})
\end{enumerate}
\caption{Perturbed interior-point method (parameters: $\protect\crho=400\sqrt{3}\cdot\log^{1/3}W$,
$\protect\cT=3\protect\crho\log W$)\label{tab:proto-algo}}
\end{algorithm}

\subsection{Our Perturbation Technique\label{sub:Perturbing-the-problem}}

Let us start by describing our perturbation technique, which heavily
uses the structure of the $b$-matching instance obtained after applying
the reduction from Section \ref{sec:reduction}. We first show how
a perturbation is applied to an arc, then we define the set of arcs
that get perturbed in each iteration. As we will see, whenever we
perturb a particular arc (to increase its resistance) we always make
sure to perturb its ``partner'' arc, i.e., the unique arc sharing
a common vertex from the set $Q$ of the bipartition, as well. 
\begin{defn}
\label{def:partner}Given an arc $e=(u,v)$, with $u\in P$, $v\in Q$,
the \textit{partner arc} of $e$ is the unique arc $\bar{e}=(\bar{u},v)$,
$\bar{u}\in P$ sharing vertex $v$ with $e$.
\end{defn}

\subsubsection{\label{sub:Perturbing-an-arc}Perturbing an Arc.}

Let $e=(u,v)$ be an arc with cost $c_{uv}$ and vertex potentials
$y_{u}$, respectively $y_{v}$, and slack $s_{uv}=c_{uv}+y_{u}-y_{v}.$
Note that due to the structure of our $b$-matching instance (see
Section \ref{sec:reduction}), vertex $v$ is of degree 2. Let $\overline{e}=(\overline{u},v)$
be the partner arc that shares with $e$ this vertex $v$. We first
modify our dual solution by setting $y_{v}\leftarrow y_{v}-s_{uv}$.
This effectively doubles the resistance of $e$, defined as in (\ref{eq:def_r_e}),
which is our desired effect. 

Unfortunately, this update breaks centrality of \emph{both }arc $e$
and its partner arc $\bar{e}$. To counteract that, we first double
the measure $\nu_{e}$ of \textbf{$e$ }- this immediately restores
the centrality of that arc. Now, it remains to fix the centrality
of the partner arc $\overline{e}=(\overline{u},v)$. Specifically,
we need to deal with the fact that the slack $s_{\overline{e}}$ of
that partner arc gets increased by $s_{e}$. To fix this problem,
recall that the centrality condition for $\overline{e}$ guaranteed
that $s_{\overline{e}}f_{\overline{e}}=\nu_{\overline{e}}\hat{\mu}$.
So, we need to set the new measure $\nu_{\overline{e}}'$ such that
$(s_{\overline{e}}+s_{e})f_{\overline{e}}=\nu_{\overline{e}}'\hat{\mu}$.
Therefore we just set the new measure to be

\[
\nu_{\overline{e}}'=\frac{(s_{\overline{e}}+s_{e})f_{\overline{e}}}{\hat{\mu}}=\nu_{\overline{e}}+\frac{s_{e}f_{\overline{e}}}{\hat{\mu}}=\nu_{\overline{e}}+\nu_{e}\cdot\frac{f_{\overline{e}}}{f_{e}}
\]
Consequently, the total change in measure of that arc is 
\begin{equation}
\nu_{e}\left(1+\frac{f_{\overline{e}}}{f_{e}}\right)\leq\nu_{e}\left(1+\frac{1}{f_{e}}\right)\label{eq:measure-increase}
\end{equation}
 as in our instance we have that $f_{\overline{e}}\leq1$, since the
arcs are unit capacitated, and $f$ is always feasible.

We remark that we may want to perturb both an arc $e=(u,v)$ and its
partner $\bar{e}=(\bar{u},v)$. In this case, we can perturb the arcs
simultaneously by setting $y_{v}\leftarrow y_{v}-s_{uv}-s_{\bar{u}v}$,
and updating the measures: $\nu_{\bar{e}}\leftarrow2\nu_{\bar{e}}+\nu_{e}\cdot\frac{f_{\bar{e}}}{f_{e}}$,
$\nu_{e}\leftarrow2\nu_{e}+\nu_{\bar{e}}\cdot\frac{f_{e}}{f_{\bar{e}}}$.
This maintains centrality, and the bound from (\ref{eq:measure-increase})
still holds.

So, to summarize, one effect of the above operation is that it made
the resistance of the perturbed arc \textbf{\emph{$e$ }}double. As
we will see, similarily as it was the case in \cite{Madry13}, this
will enable us to ensure that the total number of perturbation steps
is not too large. Also, note that the above operation does not change
any vertex demands or costs. It only affects the dual solution and
the arcs' measure. Therefore, the only undesirable long term effect
of it is the measure increase, since it might lead to violation of
Invariant 1. \footnote{In fact, if we were to formulate our problem as a primal interior-point
method, one could think of these perturbations on arcs and their partners
as acting on both the lower and the upper barriers. In that formulation,
the barrier would be of the form $-\sum_{e}\nu_{e}\log f_{e}+\nu_{\bar{e}}\log(1-f_{e})$.
The reduction from Section \ref{sec:reduction} essentially eliminates
the upper barrier, in order to be make our problem amenable to a primal-dual
approach, which we preferred to use here.}

\subsubsection{Which Arcs to Perturb?}

As we have just seen, while perturbing an arc doubles its resistance,
this operation has the side effect of increasing the total measure.
To control the latter, undesirable effect, we show that every time
we need to pertrub the problem, we can actually select a subset of
arcs with the property that perturbing all of them increases the energy
by a lot while keeping the measure increase bounded. Ultimately, the
entire goal of the analysis will be to show that:
\begin{enumerate}
\item We do not need to perturb the problem more than $\tilde{O}(m^{1/2-\eta})$
times
\item The total increase in measure caused by the perturbations is at most
$2m$, thus maintaing Invariant 1.
\end{enumerate}
Below we define the subset of arcs that will get perturbed. Intuitively,
we only want to perturb arcs $e$ with high congestion $\rho_{e}$.
Furthermore, we choose to perturb them only if the amount of flow
they carry is not very small. This extra restriction enforces an upper
bound on the amount by which the measure of the perturbed arc increases,
as per equation \ref{eq:measure-increase}.

As we will see in Corollary \ref{cor:measure-increase-bounded-by-energy},
perturbing edge $e$ will increase total energy by at least a quantity
that is proportional to the amount of contribution to total energy
of that edge; therefore the total measure increase will be upper bounded
by a quantity proportional the total energy increase due to perturbations.

We will soon see that the total energy increase across iterations
is, as a matter of fact, bounded by $O\left(\cT\crho^{2}\cdot m^{3/2-3\eta}\right)$,
where $\cT$ and $\crho$ are some appropriately chosen constants,
which immediately yields the desired bound on the total measure increase.
\begin{defn}
An arc $e$ is \textit{perturbable} if $\frac{1}{f_{e}}\leq m^{1/2-3\eta}$
and $\rho_{e}\geq\sqrt{40\cT\crho^{2}}\cdot m^{1/2-3\eta}$. An arc
that is not perturbable is called \textit{unperturbable.} Denote by
$\mathcal{S}$ the set of all perturbable arcs.\label{def:perturbable}
\end{defn}
A useful property of perturbable arcs the way they are chosen enforces
a small increase in measure compared to that in energy during each
perturbation. We will make this property precise below, and it is
what we will be using for the remainder of the section.
\begin{cor}
A perturbable arc satisfies $1+\frac{1}{f_{e}}\leq C\rho_{e}^{2}$,
where $C=\frac{1}{20\cT\crho^{2}}\cdot m^{-1/2+3\eta}$.
\end{cor}

\subsection{Runtime Analysis}

The analysis of our algorithm is based on two major parts. 

The first part shows that throughout the execution of the algorithm,
the total energy increase caused by perturbations can not be too large.
This will automatically imply that total measure increase will be
bounded by $2m$, and therefore Invariant 1 is preserved. The key
idea is that, since they are applied only when the $\ell_{3}$ norm
of the congestion is ``small'' (i.e. $\crho\cdot m^{1/2-\eta}$),
progress steps do not decrease the energy by a lot (i.e. $O(\crho^{2}\cdot m^{1-2\eta})$
, as we will soon see). However, since total measure, and hence energy,
was $O(m)$ to begin with, perturbations could not have increased
energy by more than progress steps have decreased it overall. Over
the $O\left(\cT\cdot m^{1/2-\eta}\right)$ iterations, progress steps
decrease energy by at most $O(\cT\crho^{2}\cdot m^{3/2-3\eta})$;
therefore this is also a bound on the total increase in energy. 

The second part use the invariant that perturbable edges consume most
of the energy in the graph, in order to argue that the number of perturbations
is small. While a priori we only had a bound on the time required
for progress steps, with no guarantee on how many iterations the algorithm
spends performing perturbations, this argument provides a bound on
the number of perturbations, and hence on the running time of the
algorithm. Showing that, every time we perform a perturbation, energy
increases by at least $\Omega\left(\crho\cdot m^{1-2\eta}\right)$
implies, together with the bound proven in the first part, that throughout
the execution of the algorithm we perform only $O\left(\cT\crho\cdot m^{1/2-\eta}\right)$
perturbations. This bounds the running time by $\tilde{O}\left(\cT\crho\cdot m^{3/2-\eta}\right)$,
since perturbing the problem takes only $\tilde{O}(m)$ time.

The invariant that the second part relies on is motivated by the fact
that, whenever we have to perturb the problem, the $\ell_{3}$ norm
of the congestion vector is large, so the energy of the system is
also large (at least $\crho^{2}\cdot m^{1-2\eta}$). Since perturbable
arcs are highly congested, we expect them to contribute most of the
energy; so perturbing those should increase the energy of the system
by a quantity proportional to their current contribution to energy.
Maintaining this invariant requires a finer control over how the electrical
flows behave, and will be guaranteed via a modification of the algorithm,
which will be carefully analyzed in Section \ref{sec:preconditioning}.
However, the future modification will not affect any of the analysis
described in this section. Therefore this section will be concerned
only with proving the runtime guarantee, assuming validity of the
invariant.

\subsubsection{Bounding the Total Increase in Measure and Energy}

We formalize the intuition described at the beginning of the section.
First, we show that Theorem \ref{thm:l_4_improvement_step} provides
a bound on how much energy can decrease during one progress step.
This relies on the following lemma, which allows us to lower bound
the energy of an electrical flow.
\begin{lem}
\label{lem:energy-lb}Let $\mathcal{E}_{r}$ be the energy of the
electrical flow in a graph with demands $\sigma$ and resistances
$r$. Then 
\begin{equation}
\mathcal{E}_{r}=\max_{\phi}\left(2\sigma^{\top}\phi-\sum_{e=(u,v)}\frac{\left(\phi_{u}-\phi_{v}\right)^{2}}{r_{e}}\right)\label{eq:energy-as-maximization}
\end{equation}
\end{lem}
\begin{proof}
The result can be derived by letting $L$ be the Laplacian corresponding
to the graph with resistances $r$, and rewriting the above maximization
problem as $\max_{\phi}2\sigma^{\top}\phi-\phi^{\top}L\phi$ . By
first order optimality conditions we get that the maximizer satisfies
$L\phi=\sigma$, hence $\phi=L^{+}\sigma$. Plugging in makes the
expression equal to $\sigma^{T}L^{+}\sigma$, which is precisely the
energy $\mathcal{E}_{r}$.
\end{proof}
In our context, this lemma enables us to provide a more convenient
formula for lower bounding the new value of energy after resistances
change. 
\begin{lem}
\label{lem:better-energy-lb}Let $\mathcal{E}_{r}$ be the energy
of the electrical flow corresponding to a centered instance with resistnaces
$r$, and let $\mathcal{E}_{r'}$ be the energy of the electrical
flow corresponding to the new centered instance with resistances $r'$,
obtained after applying one progress step or one perturbation. Then
the change in energy can be lower bounded by:

\begin{equation}
\mathcal{E}_{r'}-\mathcal{E}_{r}\geq\sum_{e=(u,v)}\nu_{e}\rho_{e}^{2}(1-r_{e}/r_{e'})\label{eq:energy-change-resistances}
\end{equation}
\end{lem}
\begin{proof}
Let $\hat{\phi}$ be the potentials that maximize the expression from
(\ref{eq:energy-as-maximization}) for resistances $\mathcal{E}_{r}$.
Therefore we have

\[
\mathcal{E}_{r}=2\sigma^{\top}\hat{\phi}-\sum_{e=(u,v)}\frac{\left(\hat{\phi}_{u}-\hat{\phi}_{v}\right)^{2}}{r_{e}}
\]

Using the same set of potentials in order to certify a lower bound
on the new energy, we obtain:

\begin{align*}
\mathcal{E}_{r'} & \geq2\sigma^{\top}\hat{\phi}-\sum_{e=(u,v)}\frac{\left(\hat{\phi}_{u}-\hat{\phi}_{v}\right)^{2}}{r_{e}'}\\
 & =2\sigma^{\top}\hat{\phi}-\sum_{e=(u,v)}\frac{\left(\hat{\phi}_{u}-\hat{\phi}_{v}\right)^{2}}{r_{e}}+\sum_{e=(u,v)}\frac{\left(\hat{\phi}_{u}-\hat{\phi}_{v}\right)^{2}}{r_{e}}\left(1-\frac{r_{e}}{r_{e}'}\right)\\
 & =\mathcal{E}_{r}+\sum_{e=(u,v)}\frac{\left(\hat{\phi}_{u}-\hat{\phi}_{v}\right)^{2}}{r_{e}}\left(1-\frac{r_{e}}{r_{e}'}\right)\\
 & =\mathcal{E}_{r}+\sum_{e=(u,v)}\nu_{e}\rho_{e}^{2}\left(1-\frac{r_{e}}{r_{e}'}\right)
\end{align*}

For the last identity we used the fact that if $\hat{f}$ is the eletrical
flow corresponding to potentials $\hat{\phi}$, then $\frac{\left(\hat{\phi}_{u}-\hat{\phi}_{v}\right)^{2}}{r_{e}}=r_{e}\hat{f}_{e}^{2}=\frac{\nu_{e}}{f_{e}^{2}}\hat{f}_{e}^{2}=\nu_{e}\rho_{e}^{2}$.
\end{proof}
A first application of this lemma is that it enables us to lower bound
the increase in energy when perturbing arcs.
\begin{lem}
\label{lem:energy-increase-lb}After perturbing arcs in $\mathcal{S}$,
energy increases by at least $\frac{1}{2}\norm{\rho_{\mathcal{S}}}_{\nu,2}^{2}$.\end{lem}
\begin{proof}
According to the effects of the perturbation described in Section
\ref{sub:Perturbing-an-arc}, all resistances of arcs in $\mathcal{S}$
get doubled, while the others can only increase. Therefore, applying
Lemma \ref{lem:better-energy-lb}, we obtain a lower bound on the
energy increase:
\begin{align*}
\mathcal{E}_{r'}-\mathcal{E}_{r} & \geq\sum_{e\in\mathcal{S}}\nu_{e}\rho_{e}^{2}\left(1-\frac{r_{e}}{2r_{e}}\right)\\
 & =\frac{1}{2}\sum_{e\in\mathcal{S}}\nu_{e}\rho_{e}^{2}\\
 & =\frac{1}{2}\norm{\rho_{\mathcal{S}}}_{\nu,2}^{2}
\end{align*}

\end{proof}
An immediate corollary is that the increase in energy during a perturbation
upper bounds the increase in measure.
\begin{cor}
\label{cor:measure-increase-bounded-by-energy}If a perturbation increases
energy by $\Delta$, then the total measure increases by at most $2C\cdot\Delta$.\end{cor}
\begin{proof}
By definition, the arcs we perturb satisfy $1+\frac{1}{f_{e}}\leq C\rho_{e}^{2}$.
According to (\ref{eq:measure-increase}), the measure increase cause
by perturbing an arc $e$ is at most $\nu_{e}\left(1+\frac{1}{f_{e}}\right)$.
Therefore, perturbing all the arcs in $\hat{\mathcal{S}}$, increases
measure by at most $\sum_{e\in\hat{\mathcal{S}}}\nu_{e}\cdot C\rho_{e}^{2}=C\cdot\norm{\rho_{\mathcal{S}}}_{\nu,2}^{2}$.
But Lemma \ref{lem:energy-increase-lb} shows that $\Delta\geq\frac{1}{2}\norm{\rho_{\mathcal{S}}}_{\nu,2}^{2}$.
Combining these two bounds yields the result.
\end{proof}
While Lemma \ref{lem:energy-increase-lb} tells us that perturbations
increase energy, we can show that progress steps do not decrease it
by too much, using another application of Lemma \ref{lem:better-energy-lb}.
\begin{lem}
\label{lem:progress-energy-decrease}Let $\mathcal{E}_{r}$ be the
energy of an electrical flow corresponding to a centered solution
with congestion vector $\rho$, and let $\mathcal{E}_{r'}$ be the
new energy after applying one progress step. Then $\mathcal{E}_{r'}\geq\mathcal{E}_{r}-5\cdot\norm{\rho}_{\nu,3}^{2}$.\end{lem}
\begin{proof}
Combining Lemma \ref{lem:better-energy-lb} with Theorem \ref{thm:l_4_improvement_step}
we obtain:

\begin{align*}
\mathcal{E}_{r'}-\mathcal{E}_{r} & \geq\sum_{e=(u,v)}\nu_{e}\rho_{e}^{2}\left(1-\frac{r_{e}}{r_{e}'}\right)\\
 & \geq\sum_{e=(u,v)}\nu_{e}\rho_{e}^{2}\left(-4\delta\rho_{e}-\kappa_{e}\right)\\
 & =-4\delta\cdot\sum_{e=(u,v)}\nu_{e}\rho_{e}^{3}-\sum_{e=(u,v)}\nu_{e}\rho_{e}^{2}\kappa_{e}\\
 & =-4\delta\cdot\norm{\rho}_{\nu.3}^{3}-\sum_{e=(u,v)}\sqrt{\nu_{e}}\rho_{e}^{2}\cdot\sqrt{\nu_{e}}\kappa_{e}\\
 & \geq-4\frac{\norm{\rho}_{\nu,3}^{3}}{\norm{\rho}_{\nu,4}}-\sqrt{\left(\sum_{e=(u,,v)}\nu_{e}\rho_{e}^{4}\right)\left(\sum_{e=(u,v)}\nu_{e}\kappa_{e}^{2}\right)} & \textnormal{(by Cauchy-Schwarz)}\\
 & \geq-4\norm{\rho}_{\nu,3}^{2}-\sqrt{\norm{\rho}_{\nu,4}^{4}\cdot\norm{\kappa}_{\nu,2}^{2}}\\
 & =-4\norm{\rho}_{\nu,3}^{2}-\norm{\rho}_{\nu,4}^{2} & \textnormal{(using \ensuremath{\norm{\kappa}_{\nu,2}^{2}\leq1})}\\
 & \geq-5\norm{\rho}_{\nu,3}^{2}
\end{align*}

With this tool in hand we can now upper bound the total energy increase
caused by perturbations.\end{proof}
\begin{lem}
\label{lem:total-energy-measure-increase}The total energy increase
due to perturbations is at most $16\cT\crho^{2}\cdot m^{3/2-3\eta}$.
Furthermore, the total measure always satisfies $\norm{\nu}_{1}\leq3m$,
i.e. Invariant 1 is preserved.\footnote{While the constants provided here are worse than those seen in the
proof, we will use these loose bounds in order to accomodate some
future changes in the algorithm.}\end{lem}
\begin{proof}
We start by introducing some notation. Let $\mathcal{E}^{t}$ and
$\nu^{t}$ be the energy, respectively the vectore of mesures at the
end of the $t^{\mbox{th}}$ iteration. Also, let $\Delta^{t}$ be
the total amount of energy increases during that iteration.\footnote{Remember that energy can decrease during progress steps, as per Lemma
\ref{lem:progress-energy-decrease}; $\Delta^{t}$ measures the total
amount of all increases, without accounting for the lost energy due
to progress steps. }

Note that, since we only perform progress steps when $\norm{\rho}_{\nu,3}\leq\crho m^{1/2-\eta}$,
one progress step decreases energy by at most $5\cdot\crho^{2}m^{1-2\eta}$,
according to Lemma \ref{lem:progress-energy-decrease}. Therefore
the amount by which energy increases during an iteration can be bounded
by
\begin{equation}
\Delta^{t}\leq\mathcal{E}^{t}-\mathcal{E}^{t-1}+5\crho^{2}\cdot m^{1-2\eta}\label{eq:6.4}
\end{equation}

At any point, the energy is capped by the total measure (Lemma \ref{lem:rho_l_2_norm_bound}).
Therefore 
\begin{equation}
\mathcal{E}^{t}\leq\norm{\nu^{t}}_{1}\label{eq:6.5}
\end{equation}

Also, using Corollary \ref{cor:measure-increase-bounded-by-energy}
we get that every increase in energy by $\Delta^{t}$ increases the
total measure by at most $2C\cdot\Delta^{t}$. Hence 
\begin{equation}
\norm{\nu^{t}}_{1}\leq\norm{\nu^{t-1}}_{1}+2C\cdot\Delta^{t}\label{eq:6.6}
\end{equation}

Using (\ref{eq:6.4}) and summing over all $T=\cT\cdot m^{1/2-\eta}$
iterations of the algorithm we obtain: 
\begin{align}
\sum_{t=1}^{T}\Delta^{t} & \leq\sum_{t=1}^{T}\left(\mathcal{E}^{t}-\mathcal{E}^{t-1}+5\crho^{2}\cdot m^{1-2\eta}\right)\nonumber \\
 & \leq\mathcal{E}^{T}+T\cdot5\crho^{2}\cdot m^{1-2\eta}\nonumber \\
 & \leq\norm{\nu^{T}}_{1}+T\cdot5\crho^{2}\cdot m^{1-2\eta}\nonumber \\
 & \leq\left(\norm{\nu^{0}}_{1}+\sum_{t=1}^{T}2C\cdot\Delta^{t}\right)+T\cdot5\crho^{2}\cdot m^{1-2\eta}\nonumber \\
 & =m+2C\cdot\sum_{t=1}^{T}\Delta^{t}+T\cdot5\crho^{2}\cdot m^{1-2\eta}\label{eq:energy-increase}
\end{align}

where for the last two inequalities we applied (\ref{eq:6.5}) and
(\ref{eq:6.6}).

Hence we obtain
\begin{equation}
\sum_{t=1}^{T}\Delta^{t}\leq\frac{1}{1-2C}\left(m+T\cdot5\crho^{2}\cdot m^{1-2\eta}\right)\leq2\cdot\cT m^{1/2-\eta}\cdot5\crho^{2}m^{1-2\eta}=10\cT\crho^{2}\cdot m^{3/2-3\eta}\label{eq:total-energy-increase}
\end{equation}

and the measure increase is upper bounded by 
\begin{equation}
2C\cdot\left(\sum_{t=1}^{T}\Delta^{t}\right)=2C\cdot10\cT\crho^{2}\cdot m^{3/2-3\eta}=m
\end{equation}

So Invariant 1 is satisfied. 
\end{proof}

\subsubsection{Bounding the Number of Perturbations}

We have just seen that the energy increase suffered due to perturbations
is $O(\cT\crho^{2}\cdot m^{3/2-3\eta})$, which should intuitively
enable us to bound the number of perturbations, and thus wrap up the
analysis of the algorithm. The reason is that whenever we have to
perturb the problem, the $\ell_{3}$ norm of the congestion vector
is large (i.e. $\norm{\rho}_{\nu,3}\geq\crho\cdot m^{1/2-\eta}$),
so the energy of the system is large: $\mathcal{E}=\norm{\rho}_{\nu.2}^{2}\geq\norm{\rho}_{\nu,3}^{2}\geq\crho^{2}\cdot m^{1-2\eta}$.
Since perturbable arcs are highly congested (see Definition \ref{def:perturbable}),
we expect them to contribute most of the energy. This feature of perturbable
args is highlighted by the following invariant:

\vspace{10bp}

\textbf{Invariant 2.} Whenever we perform a perturbation, $\norm{\rho_{\mathcal{S}}}_{\nu,2}^{2}\geq\frac{1}{2}\crho\cdot m^{1-2\eta}$.

\vspace{10bp}

This guarantees that every perturbation increases energy by at least
$\Omega\left(m^{1-2\eta}\right)$, which automatically implies that
the number of perturbations is bounded by $O\left(\cT\crho^{2}\cdot m^{1/2-\eta}\right)$.
Indeed, as seen in Lemma \ref{lem:energy-increase-lb}, with every
perturbation energy increases by $\frac{1}{2}\norm{\rho_{\mathcal{S}}}_{\nu,2}^{2}$.
Therefore, assuming Invariant 2, we get that each perturbation increases
energy by at least $\frac{1}{4}\crho\cdot m^{1-2\eta}$. Since we
know from Lemma \ref{lem:total-energy-measure-increase} that total
energy increase is bounded by $16\cT\crho^{2}\cdot m^{3/2-3\eta}$,
we get that the number of perturbations performed during the execution
of the algorithm is at most $64\cT\crho\cdot m^{1/2-\eta}$. 

Enforcing the validity of this invariant will be done in Section \ref{sec:preconditioning},
where we introduce a preconditioning technique which enables us to
gain more control over the behavior of electrical flows.

Hence we have proved the following Lemma:
\begin{lem}
\label{lem:number-of-perturbations}Assuming Invariant 2 is valid,
the number of perturbations is at most $64\cT\crho\cdot m^{1/2-\eta}$.
\end{lem}
This immediately concludes the running time analysis. Indeed, both
progress steps and perturbations can be implemented in $\tilde{O}(m)$
time by computing electrical flows using a fast Laplacian solver (see
Theorem \ref{thm:vanilla_SDD_solver}). The number of progress steps
is precisely $\cT\cdot m^{1/2-\eta}$, since this is hard coded in
the description of the algorithm. Also, the number of perturbations
is $O(\cT\crho\cdot m^{1/2-\eta})$, according to Lemma \ref{lem:number-of-perturbations}.
Therefore the total running time is $\tilde{O}\left(\cT\crho\cdot m^{3/2-\eta}\right)$.
\begin{thm}
\label{thm:Asumming-Invariant-2}Asumming Invariant 2 is valid, we
can produce an exact solution to the unit-capacitated minimum cost
$\sigma$-flow problem in $\tilde{O}\left(m^{10/7}\log^{4/3}W\right)$
time.\end{thm}
\begin{proof}
The proof is similar to the one for Theorem \ref{thm:min-cost-flow-m32}. 

The algorithm performs a progress step only when $\norm{\rho}_{\nu,4}\leq\norm{\rho}_{\nu,3}\leq m^{1/2-\eta}$,
therefore $\hat{\mu}$ decreases by a factor of $1-\frac{1}{\crho\cdot m^{1/2-\eta}}$
with every iteration. Therefore, in $\crho m^{1/2-\eta}\left(2\log m+\log\tilde{W}\right)\leq\cT\cdot m^{1/2-\eta}$
iterations we reduce $\hat{\mu}$ to $O(m^{-3})$, and by Invariant
1 the duality gap of this solution will be $\sum_{e}\nu_{e}\hat{\mu}\leq m^{-2}$.
Each of the $\cT\cdot m^{1/2-\eta}$ progress steps requires $\tilde{O}(1)$
electrical flow computations, and each of them can be implemented
in near-linear time, according to Theorem \ref{thm:vanilla_SDD_solver}.
Furthermore, assuming Invariant 2, we have that the number of perturbations
is at most $64\cT\crho\cdot m^{1/2-\eta}$, by Lemma \ref{lem:number-of-perturbations}.
Similarly, each perturbation can be implemented in nearly-linear time.
Therefore the total running time required for obtaining a duality
gap of $m^{-2}$ is $\tilde{O}\left(\cT m^{3/2-\eta}+\cT\crho\cdot m^{3/2-\eta}\right)=\tilde{O}\left(m^{10/7}\log^{4/3}W\right)$.
Then, using the repairing algorithm from Section \ref{sec:repair},
we can round the solution to to an optimal one in nearly-linear time.
So the total time is $\tilde{O}\left(m^{10/7}\log^{4/3}W\right)$.
\end{proof}
One should note that the $\tilde{O}\left(m^{10/7}\log^{4/3}W\right)$
running time can be reduced to $\tilde{O}\left(m^{10/7}\log W\right)$
by employing the scaling technique of \cite{gabow-85}. Thus, we can
reduce our problem to solving $O(\log W)$ instances of our problem
where the costs are polynomially bounded. This enables us to change
the $\textnormal{poly}\log W$ factors from the running time to $\textnormal{poly}\log n$
(and thus have them absorbed by the $\tilde{O}$ notation) at the
cost of paying only an extra factor of $\log W$.

However, ensuring that Invariant 2 always holds is a bit more subtle.
Obtaining a provable guarantee will actually be done by adding a preconditioner,
which is carefully analyzed in Section \ref{sec:preconditioning}.

\section{\label{sec:preconditioning}Preconditioning the Graph}

Our analysis from the previous section was crucially relying on the
assumption that perturbable arcs contribute most of the energy. Unfortunately,
this assumption is not always valid. To cope with this problem, we
develop a modification of our algorithm that ensures that this assumption
holds after all. Roughly speaking, we achieve that by an appropriate
preconditioning of our graph. This preconditioning is based on augmenting
the graph with additional, auxiliary edges which make the computed
electrical flows better behaved. These edges \textit{should not} be
thought of as being part of the graph we are computing our $\sigma$-flows
on. Their sole effect is to guide the electrical flow computation
in order to obtain a better electrical flow in the original graph
at the cost of slightly changing the demand we are routing.

These edges achieve the optimal trade-off between providing good connectivity
in the augmented graph (which lowers the congestion of arcs with low
residual capacity) and preventing too much flow from going through
them (because of their sufficiently high resistance).

One difficulty posed by this scheme is that we need to control how
much the routed demand gets modified. This is easily controled by
setting the resistances of the auxiliary edges to be sufficiently
high; in contrast, the magnitude of these resistances needs to be
traded against the effect they have on the computed electrical flow.
At the end, we fix the demand using a combinatorial procedure (see
Section \ref{sec:repair}) whose running time is proportional to the
$\ell_{1}$ difference between the initial demand and the one routed
by the algorithm. Therefore we need to simultaneously ensure that
preconditioner edges have sufficiently high resistance such that the
change in demand is not significant, and guarantee that the graph
has good enough connectivity for Invariant 2 to hold. This trade-off
will ultimately determine the choice of the parameter $\eta=1/14$.

\subsection{Using Auxiliary Edges for Electrical Flow Computations}

In order to properly describe preconditioning, we need to partition
the iterations of the algorithm into \textit{phases} (each of them
consisting of a sequence of $m^{2\eta}$ iterations), and show that
a finer version of Lemma \ref{lem:total-energy-measure-increase}
holds, for each of these phases. The reason is that the resistances
on the auxiliary edges need to depend on the set of measures in the
graph. But measures increase over time, so the resistances need to
be updated accordingly. One should be careful about this aspect, since
changing the resistances of auxiliary edges during every iteration
of the algorithm would destroy the potential based argument we described
in Section \ref{sec:longer_steps}. Therefore, instead of adjusting
the resistances every iteration, we do this only at the beginning
of a phase. Over the course of a phase, measures can not change too
much, so the preconditioning argument will still be valid. 
\begin{defn}
Partition the $\cT\cdot m^{1/2-\eta}$ iterations of the algorithm
into consecutive blocks of size $m^{2\eta}$. Such a block is called
a phase. Hence the algorithm consists of $\cT\cdot m^{1/2-3\eta}$
phases.
\end{defn}
Preconditioning consists of adding an extra vertex $v_{0}$ along
with undirected edges $(v_{0},v)$ for each vertex $v\in P$ (recall
that vertices in $P$ correspond to vertices from only one side of
the bipartition in the $b$-matching instance). We will call these
newly added edges \textit{auxiliary edges}. Each of these auxiliary
edges will have resistance set to

\begin{equation}
r_{v_{0}v}=\frac{m^{1+2\eta}}{a(v)}\label{eq:prec-resist}
\end{equation}

where we define
\begin{equation}
a(v)=\sum_{u\in Q:e=(v,u)\in E}\nu_{e}+\nu_{\bar{e}}
\end{equation}

Recall that by $\bar{e}$ we denote the partner arc of $e$ (introduced
in Definition \ref{def:partner}), and that the quantities above are
defined with respect to the measures existing at the beginning of
the phase.

Also, remember that these auxiliary edges exist only in order to provide
a mildly different demand for which electrical flows are better behaved.
Once we are done perturbing, we perform a progress step on the graph
without auxiliary edges, but with the modified demand (i.e. the one
that gets routed on the original graph, after removing auxiliary edges).

The inclusion of auxiliary edges in the electrical flow computations
requires the contribution of these edges to the energy of the system
to be included in the potential based argument from Theorem \ref{lem:total-energy-measure-increase},
when analyzing the preconditioned steps.\footnote{The $\ell_{3}$ norm of the congestion vector will still be measured
only with respect to the arcs in the original graph.} 
\begin{rem}
Even though we include additional edges for electrical flow computations,
the energy bound from Lemma \ref{lem:rho_l_2_norm_bound} still holds
(since including additional edges can only decrease energy).
\end{rem}
This motivates partitioning the iterations into phases, since changing
the resistances of auxiliary edges too often could potentially make
the energy vary wildly.

The following lemma shows that we can individually bound the energy
and measure increase over any single phase. What is crucial about
the new proof is that it does not require any control over how energy
changes between iterations belonging to different phases. Therefore,
resetting the resistances of auxiliary edges at the beginning of a
phase will have no effect on the result described in Lemma \ref{lem:total-energy-measure-increase}. 

The precise statement concerning energy and measure increase during
a phase is summarized in the following lemma, whose proof we defer
to Appendix \ref{sec:precon-proofs}.
\begin{lem}
\label{lem:total-energy-measure-increase-1}During a single phase,
the total energy increase due to perturbations is at most $16\crho^{2}\cdot m$.
Furthermore, the measure increase during a single phase is at most
$\frac{2}{\cT}\cdot m^{1/2+3\eta}$. Also, the total mesure always
satisfies $\norm{\nu}_{1}\leq3m$, i.e. Invariant 1 is preserved.
\end{lem}
It immediately follows that this is simply a refinement of Lemma \ref{lem:total-energy-measure-increase}:
\begin{cor}
The total energy increase due to perturbations is at most $16\cT\crho^{2}\cdot m^{3/2-3\eta}$,
and the result described in Lemma \ref{lem:total-energy-measure-increase}
is still valid.
\end{cor}
The new version of the algorithm which includes the effect of the
auxiliary edges is described below.

\begin{algorithm}[H]
\begin{enumerate}
\item initialize primal and dual variables $(f,y)$ (as shown in Section
\ref{sec:basic_framework})
\item repeat for $\cT\cdot m^{1/2-3\eta}$ phases
\item $\quad$reset auxiliary edge resistances (as described in (\ref{eq:prec-resist}))
\item $\quad$repeat $m^{2\eta}$ times
\item $\quad\quad$while $\Vert\rho\Vert_{\nu,3}>\crho\cdot m^{1/2-\eta}$
\item $\quad\quad\quad$perform perturbation (as shown in Section \ref{sub:Perturbing-the-problem})
\item $\quad\quad$perform progress steps on the original graph (as shown
in Section \ref{sub:Making-Progress})
\end{enumerate}
\caption{Perturbed interior-point method with preconditioning edges (parameters:
$\protect\crho=400\sqrt{3}\cdot\log^{1/3}\tilde{W}$, $\protect\cT=3\protect\crho\log\tilde{W}$)\label{tab:precond-algo}}
\end{algorithm}

Before proving that this version of the algorithm forces Invariant
2 to stay valid, we first bound the change in demand caused by the
auxiliary edges. We first show that, due to Invariant 1, the total
amount of flow that gets routed electrically through auxiliary edges
is small.
\begin{prop}
\label{prop:precond-electrical-demand}Let $\mathcal{P}$ be the set
of auxiliary edges. Then the total amount of electrical flow on these
edges during any progress step satisfies $\norm{\hat{f}_{\mathcal{P}}}_{1}\leq5\cdot m^{1/2-\eta}$.\end{prop}
\begin{proof}
From Invariant 1 and Lemma \ref{lem:rho_l_2_norm_bound} we have that
the total energy satisfies $\mathcal{E}\leq3m$. The energy contributed
by the auxiliary edges is
\begin{equation}
\sum_{(v_{0},v)\in\mathcal{P}}r_{v_{0}v}\cdot f_{v_{0}v}^{2}=\sum_{e\in\mathcal{P}}\frac{m^{1+2\eta}}{a(v)}\cdot f_{v_{0}v}^{2}\leq3m\label{eq:energy-precond}
\end{equation}

Applying Cauchy-Schwarz we obtain a bound for the $\ell_{1}$ norm
of $f_{\mathcal{P}}$:

\begin{align*}
\norm{f_{\mathcal{P}}}_{1} & =\sum_{(v_{0},v)\in\mathcal{P}}\sqrt{a(v)}\cdot\frac{\abs{f_{v_{0}v}}}{\sqrt{a(v)}}\\
 & \leq\sqrt{\left(\sum_{(v_{0},v)\in\mathcal{P}}a(v)\right)\left(\sum_{(v_{0},v)\in\mathcal{P}}\frac{f_{v_{0}v}^{2}}{a(v)}\right)}\\
 & \leq\sqrt{2\norm{\nu}_{1}\cdot\frac{3m}{m^{1+2\eta}}}\\
 & \leq\sqrt{2\cdot3m\cdot\frac{3m}{m^{1+2\eta}}}\\
 & \leq5\cdot m^{1/2-\eta}
\end{align*}

We used the fact that summing over all $a(v)$'s we obtain precisely
twice the total measure, since each measure on an arc gets counted
exactly twice. Then we used Invariant 1, and (\ref{eq:energy-precond}).
\end{proof}
This proposition shows that the demand routed within a progress step
is off by at most $5m^{1/2-\eta}$ from the demand routed by the iterate
$f$ at that point. Using this fact, we can show that the flow obtained
in the end routes a demand that is off by at most $\cT\cdot m^{1/2-\eta}$
from the original demand. 
\begin{lem}
\label{lem:Consider-the-last}Consider the last flow iterate $f^{T}$,
and let $\sigma^{T}$ be the demand routed by this flow. Then the
difference between $\sigma^{T}$ and the original demand $\sigma$
satisfies $\norm{\sigma^{T}-\sigma}_{1}\leq\cT\cdot m^{1/2-\eta}$.\end{lem}
\begin{proof}
We show by induction that after $t$ iterations, the demand routed
by $f^{t}$ satisfies $\norm{\sigma^{t}-\sigma}_{1}\leq t$. The base
case is $t=0$ where $\sigma^{0}=\sigma$, and the hypothesis holds.
Now let us show that if the induction hypothesis holds after $t-1$
iterations, then it also holds after $t$. By Proposition \ref{prop:precond-electrical-demand}
we have that the progress step first produces an electrical flow $\hat{f}^{t}$
which routes a demand $\hat{\sigma}^{t}$ on the original graph satisfying
$\norm{\hat{\sigma}^{t}-\sigma^{t-1}}_{1}\leq5m^{1/2-\eta}$. Therefore,
noting that the flow gets updated by setting it to $(1-\delta)f^{t}+\delta\elf^{t}$
for $\delta\leq\frac{1}{8\norm{\rho}_{\nu,4}}$ (see Section \ref{sub:Making-Progress}),
and using the fact that progress steps are done only when $\norm{\rho}_{\nu,3}\leq\crho\cdot m^{1/2-\eta}$
(therefore $\delta\leq\frac{1}{8\cdot\crho\cdot m^{1/2-\eta}}$),
we have that the demand $\sigma^{t}$ routed by the averaged flow
satisfies:

\begin{eqnarray*}
\norm{\sigma^{t}-\sigma}_{1} & \leq & \norm{\sigma^{t-1}-\sigma}_{1}+\norm{\sigma^{t-1}-\sigma^{t}}\\
 & = & \norm{\sigma^{t-1}-\sigma}_{1}+\norm{\sigma^{t-1}-\left((1-\delta)\sigma^{t-1}+\delta\hat{\sigma}^{t}\right)}_{1}\\
 & = & \norm{\sigma^{t-1}-\sigma}_{1}+\delta\norm{\sigma^{t-1}-\hat{\sigma}^{t}}_{1}\\
 & \leq & t-1+\delta\cdot5m^{1/2-\eta}\\
 & \leq & t-1+\frac{1}{8\cdot\crho m^{1/2-\eta}}\cdot5m^{1/2-\eta}\\
 & \leq & t-1+1\\
 & = & t
\end{eqnarray*}

Centering the solution does not change the demand, so the newly obtained
flow $f^{t}$ has the same demand $\sigma^{t}$, which satisfies the
bound above. Therefore, after $T=\cT\cdot m^{1/2-\eta}$ iterations,
we have that $\norm{\sigma^{T}-\sigma}_{1}\leq\cT\cdot m^{1/2-\eta}$.
\end{proof}

\subsection{Proving Invariant 2}

We can finally proceed with proving that, for this version of the
algorithm, Invariant 2 holds. We do so by upper bounding the $\ell_{3}$
norm of the congestions of unperturbable arcs $\norm{\rho_{\bar{\mathcal{S}}}}_{\nu,3}$.
Showing that this quantity is significantly smaller than $\tilde{O}\left(\left(\cT\crho^{2}\right)^{1/6}\cdot m^{1/2-\eta}\right)$
whenever we perform a perturbation automatically implies our result;
this is because this lower bounds the $\ell_{3}$ norm of congestions
of perturbable arcs, and therefore also their energy.
\begin{prop}
Let $\bar{\mathcal{S}}$ be the set of unperturbable arcs. If, whenever
we perform a perturbation, $\norm{\rho_{\bar{\mathcal{S}}}}_{\nu,3}\leq10(\cT\crho^{2})^{1/6}\cdot m^{1/2-\eta}$,
then Invariant 2 holds.\end{prop}
\begin{proof}
Recall that when we perform a perturbation, we have $\norm{\rho}_{\nu,3}\geq\crho\cdot m^{1/2-\eta}$.
If $\norm{\rho_{\bar{\mathcal{S}}}}_{\nu,3}\leq10(\cT\crho^{2})^{1/6}\cdot m^{1/2-\eta}$,
then 
\begin{eqnarray*}
\norm{\rho_{\mathcal{S}}}_{\nu,3} & \geq & \crho\cdot m^{1/2-\eta}-10(\cT\crho^{2})^{1/6}\cdot m^{1/2-\eta}\\
 & = & \crho\cdot m^{1/2-\eta}-10\left(3\log\tilde{W}\cdot\crho^{3}\right)^{1/6}\cdot m^{1/2-\eta}\\
 & = & \crho\cdot m^{1/2-\eta}-10\left(3\log\tilde{W}\right)^{1/6}\cdot\crho^{1/2}\cdot m^{1/2-\eta}\\
 & = & (400\sqrt{3}\log^{1/3}\tilde{W})m^{1/2-\eta}-10\cdot3^{1/6}\cdot\log^{1/6}\tilde{W}\cdot(400\sqrt{3}\log^{1/3}\tilde{W})^{1/2}\cdot m^{1/2-\eta}\\
 & = & 200\sqrt{3}\log^{1/3}\tilde{W}\cdot m^{1/2-\eta}\\
 & = & \frac{1}{2}\crho\cdot m^{1/2-\eta}
\end{eqnarray*}
 Hence $\norm{\rho_{\mathcal{S}}}_{\nu,2}^{2}\geq\norm{\rho_{\mathcal{S}}}_{\nu,3}^{2}\geq\frac{1}{2}\crho\cdot m^{1-2\eta}$.
\end{proof}
Upper bounding the $\ell_{3}$ norm of $\rho$ on unperturbable edges
is done by partitioning them into sets, and separately bounding their
$\ell_{3}^{3}$ norms. As a matter of fact, all the work we have to
do concerns arcs with congestions within the range $[\crho^{3}\cdot m^{1/2-3\eta},\sqrt{3}\cdot m^{1/2}]$,
since the energy upper bound enforced by Invariant 2 controls the
maximum congestion, while those with lower congestion immediately
satisfy the required $\ell_{3}$ norm bound.
\begin{lem}
\label{lem:bd-l3}$\norm{\rho_{\bar{\mathcal{S}}}}_{\nu,3}\leq10(\cT\crho^{2})^{1/6}\cdot m^{1/2-\eta}$.
\end{lem}
Proving Lemma \ref{lem:bd-l3} needs a careful analysis based on bounding
the contributions from edges at different scales. Doing so requires
extending the analysis from \cite{Madry13} in a slightly more complicated
fashion. That analysis looks at the line embedding of graph vertices
given by their potentials, and separately upper bounds the energy
contributed by sets of arcs according to how stretched each of them
is in the embedding. One particular obstacle posed by our setup is
that, while the analysis crucially relies on the existence of auxiliary
arcs, in our case there are no auxiliary arcs connected to vertices
in $Q$. This makes it difficult to prove statements about the amount
of energy on partner arcs $(u,v)$ and $(\bar{u},v)$ (connected to
the same vertex $v$ in $Q$), since the auxiliary edges only control
how far apart in the embedding $u$ and $\bar{u}$ are. Unfortunately,
they do not immediately say anything about $v$, which could potentially
be very far from both $u$ and $\bar{u}$, and thus the two partner
arcs could contribute a lot of energy.

We will see that our desired bound still holds. Our proof technique
relies on decomposing the electrical flow into a sum of two electrical
flows $\elfp$ and $\elfq$, one of which can be bounded directly
and the other of which has no net flow through any arc in $Q$. 

We additionally express $f=\fp+\elfq$, where we define 
\begin{equation}
\fp=f-\elf^{(Q)}
\end{equation}

The following lemma, whose proof can be found in Appendix \ref{sec:orthog-flows},
states the existence and properties of such electrical flows.
\begin{lem}
\label{lem:orthog}Let $\hat{f}$ be an electrical flow in the graph
with auxiliary edges. There exist electrical flows $\elfp$ and $\elfq$
such that:
\begin{enumerate}
\item $\elfp$ has no demand on vertices in $Q\cup\left\{ v_{0}\right\} $.
\item For any pair of partner edges $e$ and $\bar{e}$, and writing $\fp=f-\elfq$,
we have $r_{e}\left(\fp_{e}\right)^{2}+r_{\bar{e}}\left(\fp_{e}\right)^{2}\leq\nu_{e}+\nu_{\bar{e}}$
and $r_{e}\left(\elfq_{e}\right)^{2}+r_{\bar{e}}\left(\elfq\right)^{2}\leq\nu_{e}+\nu_{\bar{e}}$.
\end{enumerate}

\end{lem}
Now, we additionally define $\rhop_{e}=\frac{\abs{\elfp_{e}}}{f_{e}}$
and $\rhoq_{e}=\frac{\abs{\elfq_{e}}}{f_{e}}$. We split $\mathcal{\bar{S}}$
into two subsets: 
\[
\mathcal{A}_{1}=\left\{ e\in\mathcal{\bar{S}}\vert\rhop_{e}\geq\frac{1}{2}\rho_{e}\right\} 
\]
and 
\[
\mathcal{A}_{2}=\left\{ e\in\mathcal{\bar{S}}\vert\rhop_{e}<\frac{1}{2}\rho_{e}\right\} _{2}
\]
for which we bound the contributions to the $\ell_{3}$ norm separately.
Note that for any edge $e$ in $\mathcal{A}_{2}$, $\rhoq_{e}\geq\frac{1}{2}\rho_{e}$.
This implies that 
\begin{align*}
\norm{\rho_{\mathcal{\bar{S}}}}_{\nu,3}^{3} & \leq\norm{\rho_{\mathcal{A}_{1}}}_{\nu,3}^{3}+8\norm{\rhoq_{\mathcal{A}_{2}}}_{\nu,3}^{3}\\
 & \leq\norm{\rho_{\mathcal{A}_{1}}}_{\nu,3}^{3}+8\norm{\rhoq}_{\nu,3}^{3}
\end{align*}
First, we want to bound $\norm{\rhoq}_{\nu,3}^{3}$:
\begin{lem}
\label{lem:partner-arc-q-bound}Suppose that for all pairs of partner
edges $e$ and $\bar{e}$, $\frac{\max(\nu_{e},\nu_{\bar{e}})}{\min(\nu_{e},\nu_{\bar{e}})}\leq x$.
Then $\norm{\rhoq}_{\nu,3}^{3}\leq(3\sqrt{1+x})m$.\end{lem}
\begin{proof}
We look at the contribution of a single pair of partner arcs, $e$
and $\bar{e}$, to $\norm{\rhoq}_{\nu,3}^{3}$. This is equal to
\begin{align*}
\nu_{e}\abs{\rhoq_{e}}^{3}+\nu_{\bar{e}}\abs{\rhoq_{\bar{e}}}^{3} & =r_{e}\left(\elfq_{e}\right)^{2}\rhoq_{e}+r_{\bar{e}}\left(\elfq\right)^{2}\rhoq_{\bar{e}}\\
 & =r_{e}\left(\elfq_{e}\right)^{2}\sqrt{\frac{r_{e}\left(\elfq_{e}\right)^{2}}{\nu_{e}}}+r_{\bar{e}}\left(\elfq_{\bar{e}}\right)^{2}\sqrt{\frac{r_{\bar{e}}\left(\elfq_{\bar{e}}\right)^{2}}{\nu_{\bar{e}}}}\\
 & \leq\sqrt{\frac{\max\left(r_{e}\left(\elfq_{e}\right)^{2},r_{\bar{e}}\left(\elfq_{\bar{e}}\right)^{2}\right)}{\min(\nu_{e},\nu_{\bar{e}})}}\left(r_{e}\left(\elfq_{e}\right)^{2}+r_{\bar{e}}\left(\elfq_{\bar{e}}\right)^{2}\right)\\
 & \leq\sqrt{\frac{\nu_{e}+\nu_{\bar{e}}}{\min(\nu_{e},\nu_{\bar{e}})}}(\nu_{e}+\nu_{\bar{e}})\\
 & \leq\sqrt{1+x}(\nu_{e}+\nu_{\bar{e}}).
\end{align*}
Here, we applied property 2 from Lemma \ref{lem:orthog}. Summing
over all pairs of partner arcs and using Invariant 1, we get $\norm{\rhoq}_{\nu,3}^{3}\leq(3\sqrt{1+x})m$,
as desired.\end{proof}
\begin{cor}
$\norm{\rhoq}_{\nu,3}^{3}\leq5m^{5/4-(3/2)\eta}$\end{cor}
\begin{proof}
We can apply Lemma \ref{lem:partner-arc-q-bound} with $x=m^{1/2-3\eta}$,
since by definition we never perturb arcs with $\frac{1}{f_{e}}>m^{1/2-3\eta}$.
So $\norm{\rhoq}_{\nu,3}^{3}\leq\sqrt{1+m^{1/2-3\eta}}\cdot3m\leq5m^{5/4-(3/2)\eta}$. 
\end{proof}
It remains to bound $\norm{\rho_{\mathcal{A}_{1}}}_{\nu,3}^{3}$.
First, note that for every edge $e$ in $\mathcal{A}_{1}$, if $\rho_{e}>\sqrt{40\cT\crho^{2}}\cdot m^{1/2-3\eta}$,
we have $r_{e}=\frac{\nu_{e}}{f_{e}^{2}}$, implying that 
\begin{align*}
r_{e}\abs{\elf_{e}} & =\frac{\nu_{e}}{f_{e}}\rho_{e}\\
 & \geq\frac{1}{f_{e}}\rho_{e}\\
 & \geq m^{1/2-3\eta}\rho_{e}.
\end{align*}

By the definition of $\mathcal{A}_{1},$ $r_{e}\abs{\elfp}$ is at
least half of $\frac{1}{2}m^{1/2-3\eta}\rho_{e}$. Similarly, $r_{e}\left(\elfp_{e}\right)^{2}\geq\frac{1}{4}r_{e}f_{e}^{2}$.
Thus for any choice of a threshold $T>\sqrt{40\cT\crho^{2}}\cdot m^{1/2-3\eta}$
on the $\rho$, we have
\[
\sum_{e\in\mathcal{A}_{1},\rho_{e}\geq T}r_{e}f_{e}^{2}\leq4\sum_{r_{e}\abs{\elfp_{e}}\geq\frac{1}{2}m^{1/2-3\eta}T}r_{e}\left(\elfp_{e}\right)^{2}.
\]

Next, we consider a ``quotient'' or ``Schur complement'' electrical
network $N$ on the vertices of $P\cup\left\{ v_{0}\right\} $ only,
replacing the pair of arcs $(u,v)$ and $(\bar{u},v)$ with one edge
$(u,\bar{u)}$ with resistance $r_{uv}+r_{\bar{u}v}$. We consider
this edge to have measure $\nu_{uv}+\nu_{\bar{u}v}$. Note that this
is in some sense undoing the $b$-matching reduction from Lemma \ref{lem:reduction-to-bmatching}.

Because they have 0 net flow on all vertices in $Q$, we can map flows
$\fp$ and $\elfp$ to flows in $N$, $\fn$ and $\elfn$, by setting
the flow from $u$ to $\bar{u}$ in the new flow to that from $u$
to $v$ (or equivalently, by flow conservation, $v$ to $\bar{u}$)
in the original flow. This mapping preserves the demands on the vertices
in $Q$, so $\fn$ still satisfies the same demands as $\elfn$. It
also preserves the electrical nature of $\elfn$, since $\fn$ can
be induced by the same voltages. Finally, since the resistance on
$e=(u,\bar{u})$ is set to $r_{uv}+r_{\bar{u}v}$, we have $r_{e}\left(\fn_{e}\right)^{2}=r_{uv}\left(\fp_{uv}\right)^{2}+r_{\bar{u}v}\left(\fp_{\bar{u}v}\right)^{2}$,
and in particular for any edge $e$ in $N$ we have $r_{e}\left(\fn_{e}\right)^{2}\leq\nu_{e}$.

Now we may apply our general preconditioning result, which bounds
the total energy of edges with high voltage in terms of the resistance
of the preconditioning edges:
\begin{lem}
\label{lem:full-precon}Let $N$ be an electrical network, on a set
of vertices $P$ plus a special vertex $v_{0}$, with each edge $e$
not incident to $v_{0}$ assigned a measure $\nu_{e}.$ Let $\nu'$
be another assignment of measures such that $\nu'\leq\nu$, with the
total ``missing measure'' $\sum_{e}\nu_{e}-\nu'_{e}=M$, and let
each vertex $v$ be connected to $v_{0}$ by an edge with resistance
$\frac{R}{a(v)}$, where $a(v)$ is the sum of $\nu'_{e}$ over edges
$e$ incident to $v$. Let $\fn$ be a flow on $N$, with no flow
on any edge incident to $v_{0}$ and with $r_{e}\left(\fn\right)^{2}\leq\nu_{e}$
for every edge not incident to $v_{0}$, and let $\elfn$ be the electrical
flow on $N$ satisfying the same demands as $\fn$. Then
\[
\sum_{r_{e}\abs{\elfn_{e}}\geq V}r_{e}\left(\elfn_{e}\right)^{2}\leq\frac{32R\sum_{e}\nu_{e}}{V^{2}}+2M.
\]

\end{lem}
This is proved in Appendix \ref{sec:precon-proofs}. Here, we set
$\nu'$ to the measures from the beginning of the preconditioning
phase; by Lemma \ref{lem:total-energy-measure-increase-1} the missing
measure $M\leq\frac{2}{\cT}\cdot m^{1/2+3\eta}$. The preconditioning
edges were weighted with $R=m^{1+2\eta}$. This implies that for $T>\sqrt{40\cT\crho^{2}}\cdot m^{1/2-3\eta},$
\[
\sum_{e\in\mathcal{A}_{1},\rho_{e}\geq T}r_{e}f_{e}^{2}\leq4\cdot\left(\frac{32\cdot m^{1+2\eta}\cdot\sum_{e}\nu_{e}}{\left(\frac{1}{2}m^{1/2-3\eta}\cdot T\right)^{2}}+\frac{2}{\cT}\cdot m^{1/2+3\eta}\right)\leq\frac{512\cdot m^{8\eta}\cdot3m}{T^{2}}+\frac{8}{\cT}\cdot m^{1/2+3\eta}.
\]

On the other hand, for all $T$ we trivially have $\sum_{e\in\mathcal{A}_{1},\rho_{e}\geq T}r_{e}f_{e}^{2}\leq3m$,
by Invariant 1 combined with Lemma \ref{lem:rho_l_2_norm_bound}.
Now, we can write
\begin{align*}
\norm{\rho_{\mathcal{A}_{1}}}_{\nu,3}^{3} & =3\int_{0}^{\sqrt{3}\cdot m^{1/2}}\left(\sum_{e\in\mathcal{A}_{3},\rho_{e}\geq T}r_{e}f_{e}^{2}\right)\;dT\\
 & \leq3\left(\int_{0}^{\sqrt{40\cT\crho^{2}}\cdot m^{1/2-3\eta}}3m\;dT\right)+3\left(\int_{m^{1/2-3\eta}}^{\sqrt{3}\cdot m^{1/2}}\left(1536\cdot\frac{m^{1+8\eta}}{T^{2}}+\frac{8}{\cT}\cdot m^{1/2+3\eta}\right)\;dT\right)\\
 & \leq9\sqrt{40\cT\crho^{2}}\cdot m^{3/2-3\eta}+42m^{1+3\eta}+3\int_{m^{1/2-3\eta}}^{\infty}\left(1536\frac{m^{1+8\eta}}{T^{2}}\right)\;dT\\
 & \leq9\sqrt{40\cT\crho^{2}}\cdot m^{3/2-3\eta}+42m^{1+3\eta}+4608m^{1/2+11\eta}\\
 & \leq9\sqrt{40\cT\crho^{2}}\cdot m^{3/2-3\eta}+42m^{1+3\eta}+4608m^{1/2+11\eta}.
\end{align*}

This determines our choice of $\eta=\frac{1}{14}$--it ensures that
$m^{1/2+11\eta}=m^{3/2-3\eta}$ in the last term, which is essentially
the bound we need. Plugging that in it also ensures that $m^{1+3\eta}$
in the second term is less than $m^{3/2-3\eta}$. Finally, we had
$\norm{\rho_{\mathcal{\bar{S}}}}_{\nu,3}^{3}\leq\norm{\rho_{\mathcal{A}_{1}}}_{\nu,3}^{3}+8\norm{\rhoq}_{\nu,3}^{3}$,
with $\norm{\rhoq}\leq5m^{5/4-(3/2)\eta}$; with $\eta=\frac{1}{14}$
we also have $m^{5/4-(3/2)\eta}\leq m^{3/2-3\eta}$. Thus we have
$\norm{\rho_{\mathcal{\bar{S}}}}_{\nu,3}^{3}\leq1000\sqrt{\cT\crho^{2}}m^{3/2-3\eta}$,
or $\norm{\rho_{\mathcal{\bar{S}}}}_{\nu,3}\leq10(\cT\crho^{2})^{1/6}\cdot m^{1/2-\eta}$,
proving Lemma \ref{lem:bd-l3}.

\section{Repairing the Matching\label{sec:repair}}

In this section we assume the $b$-matching view on $\sigma$-flows.
Hence, summarizing Theorem \ref{thm:Asumming-Invariant-2}, Lemma
\ref{lem:Consider-the-last} and the proof of Invariant 2 in Section
\ref{sec:preconditioning} we obtain the following result. 
\begin{thm}
\label{thm:main} Consider a ternary instance $G=(V,E,c)$ of the
weighted perfect bipartite $b$-matching problem when $\norm b_{1}=O(m)$.
In $\tilde{O}(m^{10/7}\log^{4/3}W)$ time we can either conclude that
$G$ does not have a perfect $b$-matching or return a primal-dual
pair with duality gap at most $m^{-2}$ to the perfect $b^{+}$-matching
problem, where $\norm{b^{+}-b}_{1}\le c_{T}m^{3/7}$. 
\end{thm}
In this section we will be ``repairing'' the feasible primal-dual
solution given by the above theorem. Our repair procedure will consists
out of four steps. In the first step -- Lemma \ref{thm:rounding-first-step},
we will reduce the duality gap to 0. Next, in Lemma \ref{lem:lemma-integral-rounding}
we will round the solution to be integral. In the third step, we will
repair the perturbations done to demands in Theorem \ref{theorem:matchings}.
Finally, we will use scaling to reduce the dependence on $\log W$.

The main tool we are going to use is the directed version of the graph
bipartite graph that encodes alternating paths with respect to the
current $b$-matching. Let $G=(V_{1}\cup V_{2},E,c)$ be a bipartite
weighted graph in which we want to find a minimum weight perfect $b$-matching.
Given some fractional $b$-matching $x$ we define $\ora G_{x}=(V_{1}\cup V_{2},\ora E_{x},\ora c_{x})$
to be a directed version of the graph $G$ where all edges are directed
from $V_{1}$ to $V_{2}$ and additionally we add edges of $x$ that
are directed from $V_{2}$ to $V_{1}$ (edges of $x$ have two copies
in both directions). Moreover, the weights of edges in $x$ are negated.
Formally, 
\begin{eqnarray*}
\ora E_{x} & = & \left\{ (u,v)\vert uv\in E,u\in V_{1},v\in V_{2}\right\} \cup\left\{ (u,v)\vert x_{uv}\neq0,u\in V_{2},v\in V_{1}\right\} ,\\
\ora c_{x}(u,v) & = & \left\{ \begin{array}{ll}
c_{uv} & \textnormal{if }u\in V_{1},v\in V_{2},\\
-c_{uv} & \textnormal{if }u\in V_{2},v\in V_{1}.
\end{array}\right.
\end{eqnarray*}
We observe a path in $\ora G_{x}$ correspond to alternating paths
in $G$ with respect to $x$. Let $F_{x}$ denote the set of vertices
whose demand is not fully satisfied, i.e., $F_{x}=\left\{ v\in V\vert x(v)<b_{v}\right\} $.
We now observe that if a path $\pi$ starts in $V_{1}\cap F_{x}$
and ends in $V_{2}\cap F_{x}$ that it is an augmenting path with
respect to $x$ and can be used to enlarge the $b$-matching $x$.
When the $b$-matching $x$ is integral we can interpret it as a multiset
of edges which we denote by $M$. In such case we will use $\ora G_{M}$
to denote $\ora G_{x}$. The important property of $\ora G_{x}$ is
that given an optimal primal solution $x$ it allows to find optimal
dual solution. Let $\dist_{\ora G_{M}}(V_{1},u)$ denote the distance
from $V_{1}$ to $u$ in $\ora G_{M}$. The following property was
first observed by Iri \cite{iri60}.
\begin{lem}
Consider a ternary instance $G=(V,E,c)$ of the weighted perfect bipartite
$b$-matching problem if $x$ is an optimal $b$-matching then optimal
dual solution $y$ is given as:

\[
y_{v}=\begin{cases}
-\dist_{\ora G_{M}}(V_{1},v) & \textnormal{if }v\in V_{1},\\
\dist_{\ora G_{M}}(V_{1},v) & \textnormal{if }v\in V_{2}.
\end{cases}
\]
\end{lem}
\begin{proof}
First, observe that $\ora G_{x}$ does not contain negative length
cycles by optimality of $x$, so the above distances are well defined.
Second, observe that $y$ is feasible as for each $uv\in E$, where
$u\in V_{1}$ and $v\in V_{2}$ we have $\dist_{\ora G_{M}}(V_{1},v)\le\dist_{\ora G_{M}}(V_{1},u)+c_{uv}$,
so $y_{v}+y_{u}\le c_{uv}.$ Third, when $x_{\ensuremath{uv}}\neq0$
both edges $(u,v)$ and $(v,u)$ are present in $\ora G_{M}$. One
edge implies feasibility $y_{v}+y_{u}\le c_{uv}$, whereas in the
second one all signs are reversed $-y_{v}-y_{u}\le-c_{uv}$. Hence,
the equality $y_{v}+y_{u}=c_{uv}$ and $y$ is optimal.
\end{proof}
The above observation will be useful in the following lemma where
we will construct optimal primal-dual pair. We will first find an
optimal primal solution to the slightly perturbed instance and then
compute the corresponding optimal dual.
\begin{lem}
\label{thm:rounding-first-step} Consider a ternary instance $G=(V,E,c)$
of the weighted perfect bipartite $b$-matching problem. Given a feasible
primal-dual pair $(x,y)$ with duality gap at most $m^{-2}$ in $O(m+n\log n)$
time we can compute an optimal primal-dual pair $(x^{+},y^{+})$ to
the perfect $b^{+}$-matching problem, where $\norm{b^{+}-b}_{1}\le m^{-2}$. \end{lem}
\begin{proof}
First, we obtain $x^{+}$ from $x$ by rounding to $0$ all edges
that have value of $x$ smaller than $m^{-2}$, i.e., 

\[
x_{uv}^{+}=\begin{cases}
0 & \textnormal{if }x_{uv}\le m^{-2},\\
x_{uv} & \textnormal{otherwise.}
\end{cases}
\]

Moreover, we define $b^{+}(u)=\sum_{uv}x(v)$. Observe that all $uv\in E$
such that $x_{uv}^{+}\neq0$ we need to have $c_{uv}-y_{u}-y_{v}\le m^{-2}$,
as otherwise the duality gap would be bigger then $m^{-2}$. Consider
\emph{reduced weights} (or the slack) of edges that are defined as
$\tilde{c}_{uv}=c_{uv}-y_{u}-y_{v}$. Now define

\[
\tilde{c}_{uv}^{+}=\begin{cases}
0 & \textnormal{if }\tilde{c}_{uv}\le m^{-2},\\
\tilde{c}_{uv} & \textnormal{otherwise.}
\end{cases}
\]

We denote by $\tilde{\dist}_{\ora G_{M}}^{+}(V_{1},v)$ the distances
in $\ora G_{x}$ with respect to $\tilde{c}_{uv}^{+}$, whereas by
$\tilde{\dist}_{\ora G_{M}}(V_{1},v)$ the distances with respect
to $\tilde{c}_{uv}$. Observe that $\ora G_{x}$ with weights given
by $\tilde{c}_{uv}^{+}$ does not contain negative weight edges, so
distances with respect to $\tilde{c}_{uv}^{+}$ can be computed in
$O(m+n\log n)$ time using Dijkstra's algorithm. Moreover, we observe
that 

\[
\dist_{\ora G_{M}}(V_{1},v)=\begin{cases}
\tilde{\dist}_{\ora G_{M}}(V_{1},v)+y_{v} & \mathrm{\textnormal{if }v\in V_{1}},\\
\tilde{\dist}_{\ora G_{M}}(V_{1},v)-y_{v} & \mathrm{\textnormal{if }v\in V_{2}.}
\end{cases}
\]

We have that $\tilde{|\dist}_{\ora G_{M}}(V_{1},v)-\tilde{\dist}_{\ora G_{M}}^{+}(V_{1},v)|\le m^{-1}$
and we know that $\dist_{\ora G_{M}}(V_{1},v)$ is integral, so for 

\[
\dist_{\ora G_{M}}^{+}(V_{1},v)=\begin{cases}
\tilde{\dist}_{\ora G_{M}}^{+}(V_{1},v)+y_{v} & \mathrm{\textnormal{if }v\in V_{1}},\\
\tilde{\dist}_{\ora G_{M}}^{+}(V_{1},v)-y_{v} & \textnormal{if }\mathrm{v\in V_{2}.}
\end{cases}
\]

we have $\dist_{\ora G_{M}}(V_{1},v)=[\dist_{\ora G_{M}}^{+}(V_{1},v)]$,
where $[.]$ is the nearest integer function. 
\end{proof}
Now we are ready to round the above solution -- for the rounding step
it is essential that primal solution is optimal. The following result
that was proven recently in \cite{flowrounding} will become handy
for us.
\begin{thm}
\cite{flowrounding}\label{theorem:flowrounding} Let $N$ be a flow
network with integral capacities and edge costs. Let $f$ be a flow
in $N$ with integral total value then in $O(m\log n)$ one can compute
integral flow $f'$ with the same flow value and no worse cost. Moreover,
the support of $f'$ is a subset of the support of $f$. 
\end{thm}
The above result is obtained by performing fractional cycle cancelation
using dynamic link-cut trees~\cite{st}. 
\begin{lem}
\label{lem:lemma-integral-rounding}Let $x$ be the optimal primal
and $y$ be the optimal dual solution to the $b^{+}$-matching problem,
where $\norm{b^{+}-b}_{1}\le c_{T}\cdot m^{3/7}$. In $O(m\log n)$
time we can compute an integral $(2c_{T}\cdot m^{3/7}+1)$-near $b$-matching
$M\subseteq E$ such that $M$ is included in the support of $x$,
i.e., $M\subseteq\left\{ e\in E\vert x_{e}\neq0\right\} $. \end{lem}
\begin{proof}
For graph $G$ let $V_{1}$ and $V_{2}$ be the bipartition of $V$.
For notational convenience in the remainder of this section we use
$V_{1}=P$ and $V_{2}=Q$. First, for each vertex $v\in V$ we define
$b^{\le}$ to be $b_{v}^{\le}=\min(b_{v},b_{v}^{+})$. Moreover, we
round down $b^{\le}(V_{1})$ and $b^{\le}(V_{2})$ to the nearest
integer. Observe that $\norm{b^{+}-b}_{1}\le2c_{T}\cdot m^{3/7}+1$.
Now, for each $x(E(v))>b_{v}^{\le}$ we round down $x(E(v))$ to $b_{v}^{\le}$,
i.e., we decrease $x$ on arbitrary edges incident to $v$, so that
the fraction of edges incident to $v$ becomes $b_{v}^{\le}$. Let
us denote the resulting vector by $x^{\le}$. Now, let us view the
$b^{\le}$-matching problem as a flow problem by: 
\begin{itemize}
\item directing edge in $e\in E$ from $V_{1}$ to $V_{2}$ -- the flow
on arc $e$ is equal $x_{e}^{\le}$, 
\item adding source $s$ and sink $t$, 
\item connecting $s$ to all vertices in $v\in V_{1}$ -- the flow on arc
$sv$ is equal to $x^{\le}(E(v))$, 
\item connecting all vertices in $v\in V_{2}$ to $t$ -- the flow on arc
$vt$ is equal to $x^{\le}(E(v))$. 
\end{itemize}
By applying Theorem~\ref{theorem:flowrounding} to the above fractional
flow we obtain integral flow $f'$ with value $\ge\norm b_{1}/2-2c_{T}\cdot m^{3/7}-1$.
This flow induces $(2c_{T}\cdot m^{3/7}+1)$-near $b$-matching $M$
in $G$. 
\end{proof}
We observe that complementary slackness conditions still hold between
$M$ and $y$ (i.e., for each $uv\in M$ we have $c_{uv}=y_{u}+y_{v}$)
because $M$ is contained in the support of $x$. We will exploit
this fact in the following.

In this moment we have executed two steps of our repair procedure
-- we have reduced the duality gap to 0 and rounded the solution to
be integral. Now, we are ready to repair the perturbations done to
demands. Our repair procedure is presented in Algorithm~\ref{algorithm:matchings}.
It finds a minimum weight perfect $b$-matching in $G$. In the procedure
we first apply Theorem~\ref{thm:main} then we round the primal solution
using Theorem~\ref{theorem:flowrounding}. Next, we repeatedly find
shortest path in $\ora G_{M}$ from $V_{1}\cap F_{M}$ to $V_{2}\cap F_{M}$
with respect to reduced weights $\tilde{c}$. \emph{Reduced weights}
(or the slack) are defined as $\tilde{c}_{uv}=c_{uv}-y_{u}-y_{v}$.
These paths are used to augment the matching. Augmentation of the
matching using shortest paths guarantees that $M$ is \emph{extremal},
i.e., $M$ is the minimum weight perfect $\deg_{M}$-matching. In
$\deg_{M}$-matching the demand for vertex $v$ is equal to $\deg_{M}(v),$
i.e., number of edges incident to $v$ in $M$.\footnote{Similar definition is used in Edmonds-Karp algorithm. However, there
extremal matching is defined to be maximum weight matching of given
size.} The following corollary states that this is true at the beginning
of the algorithm.
\begin{cor}
The matching $M$ constructed by Theorem~\ref{theorem:flowrounding}
is extremal. \end{cor}
\begin{proof}
Observe that edges of $M$ are tight in $G$ with respect to the corresponding
dual solution $y$, so $\sum_{v\in V(M)}y_{v}b_{v}=c(M)$ and this
proves that $M$ is optimal $\deg_{M}$-matching. 
\end{proof}
In order to be able to efficiently find these shortest paths in $\ora G_{M}$
we need to make sure that during the execution of the algorithm reduced
weights are nonnegative. At the beginning, by the definition of $b$-vertex
packing, the reduced weights are nonnegative. However, when augmenting
the matching using some path $\pi$ we change direction of some edges
on $\pi$ in $\ora G_{M}$ and flip the sign of their weights. The
crucial part is to first reweigh the dual to make sure that all edges
on $\pi$ have reduced cost equal to $0$. If this is the case reversing
edges does not introduce negative weights. Let us denote by $R_{M}$
the set of vertices reachable from $V_{1}\cap F_{M}$ in $\ora G_{M}$
and by $\dist_{\ora G_{M}}(V_{1},u)$ the distance from $V_{1}$ to
$u$ in $\ora G_{M}$.

\begin{algorithm}
\begin{enumerate}
\item Apply Theorem~\ref{thm:main} to $G$ and let $x$ and $y$ be the
resulting primal and dual solution to the perturbed $b^{+}$.
\item Apply Theorem~\ref{theorem:flowrounding} to $x$ to obtain $(2c_{T}\cdot m^{3/7}+1)$-near
$b$-matching $M$. 
\item while $M$ is not a perfect $b$-matching repeat
\item $\quad$\emph{(Invariant: for all edges $e\in\ora G_{M}[R_{M}]$ we
have $\tilde{c}_{e}\ge0$.)}
\item $\quad$Construct $\ora G_{M}$ using $G,M,\tilde{c}$. 
\item $\quad$Find a shortest path $\pi$ from $V_{1}\cap F_{M}$ to $V_{2}\cap F_{M}$
in $\ora G_{M}$. 
\item $\quad$for all $u\in V_{1}\cup V_{2}$ do 
\item $\quad\quad$if $u$ is reachable from $V_{1}$ in $\ora G_{M}$ then
\item $\quad\quad\quad$if $u\in V_{1}$ then $y_{u}:=y_{u}-\dist_{\ora G_{M}}(V_{1},u)$
\item $\quad\quad\quad\quad\quad\quad\quad$else $y_{u}:=y_{u}+\dist_{\ora G_{M}}(V_{1},u)$ 
\item $\quad$Enlarge $M$ using augmenting path $\pi.$
\item Return $M$
\end{enumerate}
\caption{Algorithm for computing minimum weight perfect $b$-matching. }
\label{algorithm:matchings}
\end{algorithm}

First, we need to prove that the invariant in the while loop of the
algorithm holds. 
\begin{lem}
During the execution of the wile loop in Algorithm~\ref{algorithm:matchings}
for all edges $e\in\ora G_{M}[R_{M}]$ we have $\tilde{c}_{e}\ge0$.\end{lem}
\begin{proof}
First, we observe that set of reachable vertices $R_{M}$ decreases
during the execution of the algorithm. The only step that alters set
$R_{M}$ is the matching augmentation that changes the direction of
some edges on $\pi$. Hence, $R_{M}$ would increase when there would
be an edge of $\pi$ entering $R_{M}$, but this is impossible as
$\pi$ needs to be contained in $R_{M}$ by definition.

Now, we need to consider the reweighing done in the algorithm. We
will prove that when it is done we have that $\tilde{c}_{e}=0$ for
all $e\in\pi$ and $\tilde{c}_{e}\ge0$ for all $e\in\ora G_{M}[R_{M}]$.
In such case augmenting $M$ using $\pi$ will not introduce negative
weights to $\ora G_{M}$ as all weights on $\pi$ will be zero.

As for the fact that $\tilde{c}_{e}=0$ for all $e\in\pi$ we need
to consider two cases. Let $uv\in\pi$ where $u\in V_{1}$ and $v\in V_{2}$. 
\begin{itemize}
\item $uv\in M$ -- Observe that in $\ora G_{M}$ arc $vu$ enters $u$
and lies on the shortest path $\pi$, so we have $\dist_{\ora G_{M}}(V_{1},u)=\dist_{\ora G_{M}}(V_{1},v)-\tilde{c}_{uv}$,
by definition of $\ora c_{M}$. By our reweighing rule the new reduced
weight of $uv$ is equal to $\tilde{c}(uv)+\dist_{\ora G_{M}}(V_{1},u)-\dist_{\ora G_{M}}(V_{1},v)=0.$
\item $uv\not\in M$ -- As $\pi$ is the shortest path, for arc $uv$ we
have that $\dist_{\ora G_{M}}(V_{1},v)=\dist_{\ora G_{M}}(V_{1},u)+\tilde{c}_{uv}$
by definition of $\ora c_{M}$. This in turn means that the new reduced
weight of $uv$ is $\tilde{c}_{uv}+\dist_{\ora G_{M}}(V_{1},u)-\dist_{\ora G_{M}}(V_{1},v)=0$.
\end{itemize}
Let us now consider edges $e\not\in E(\pi)$ but $e\in\ora G[R_{M}]$.
We have two cases here as well. Let $uv\in E(\pi)$ where $u\in V_{1}$
and $v\in V_{2}$. 
\begin{itemize}
\item $uv\in M$ -- We have $\dist_{\ora G_{M}}(V_{1},u)\le\dist_{\ora G_{M}}(V_{1},v)-\tilde{c}_{uv}$
as well as $\dist_{\ora G_{M}}(V_{1},v)\le\dist_{\ora G_{M}}(V_{1},u)+\tilde{c}_{uv}$,
because edges of $M$ are bidirected. Hence, $\dist_{\ora G_{M}}(V_{1},v)=\dist_{\ora G_{M}}(V_{1},u)+\tilde{c}_{uv}$
and the new weight of $uv$ is $\tilde{c}_{uv}+\dist_{\ora G_{M}}(V_{1},u)-\dist_{\ora G_{M}}(V_{1},v)=0$.
\item $uv\not\in M$ -- By the properties of the distance function we have
that $\dist_{\ora G_{M}}(V_{1},v)\le\dist_{\ora G_{M}}(V_{1},u)+\tilde{c}_{uv}$,
so the new reduced weight of $uv$ is $\tilde{c}_{uv}(+\dist_{\ora G_{M}}(V_{1},u)-\dist_{\ora G_{M}}(V_{1},v)\ge0$. 
\end{itemize}
\end{proof}
We are now ready to prove the correctness of the algorithm and bound
its running time.
\begin{thm}
\label{theorem:matchings} Assuming $\norm b_{1}=O(m)$, we can find
a minimum weight perfect $b$-matching in $\tilde{O}(m^{10/7}\log^{4/3}W)$
time. \end{thm}
\begin{proof}
We apply Algorithm~\ref{algorithm:matchings}.The correctness of
the algorithm follows by the fact that after each augmentation $M$
is extremal. Hence, when $M$ is perfect it needs to be minimum cost
perfect $b$-matching.

The execution of Theorem~\ref{thm:main} requires $\tilde{O}(m^{10/7}\log^{4/3}W)$
time. In order to round $x$ to $M$ using Theorem~\ref{theorem:flowrounding}
we need $O(m\log n)$ time. Finally, as $M$ is a $\tilde{O}(m^{3/7})$-near
$b$-matching we will find at most $\tilde{O}(m^{3/7})$ augmenting
paths with respect to it. Finding each augmenting path takes $O(m+n\log n)$
time using Dijkstra's algorithm, so this part takes $\tilde{O}(m^{10/7})$
time. Therefore the total running time is $\tilde{O}\left(m^{10/7}\log^{4/3}W\right)$.\end{proof}
\begin{thm}
\label{theorem:matchings-dual} Assuming $\norm b_{1}=O(m)$, we can
find a maximum weight vertex $b$-packing in $\tilde{O}(m^{10/7}\log^{4/3}W)$
time. \end{thm}
\begin{proof}
One would like to apply Theorem~\ref{thm:main} to $G'$ to produce
the dual, but Algorithm~\ref{algorithm:matchings} does not compute
a dual solution. The problem is that $y$ in the algorithm is defined
only on the set $R_{M}$, i.e., nodes reachable in $\ora G_{M}$ from
free vertices in $V_{1}$. There is, however, an easy fix to this.
Let $v\in V_{1}\setminus V(M)$ be arbitrary free vertex in $V_{1}$.
We define $\ora G'_{M}$ to be a graph obtained from $\ora G_{M}$
by connecting $v$ with all vertices in $V_{2}$ with an edge of cost
$n\norm w_{\infty}$. Such heavy edges will never be used by a minimum
cost perfect $1$-matching, but now we always have $R_{M}=V_{1}\cup V_{2}$.
Hence, by the invariant, $y$ computed during the execution of Algorithm~\ref{algorithm:matchings}
needs to form a maximum vertex $b$-packing.
\end{proof}
In order to reduce the dependency on $\log W$ in the running time,
we apply the scaling technique of Gabow (see \cite{gabow-85}). This
enables us to reduce our problem in a black-box manner to solving
$O(\log W)$ instances where the weights are polynomially bounded,
i.e., $W\leq\norm b_{1}$. So the running time will be reduced to
$\tilde{O}(m^{10/7}\log W)$. This is described in the following theorem
that shows how to execute scaling on the dual problem.
\begin{thm}
\label{theorem:matchings-improved}Assuming $\norm b_{1}=O(m)$, we
can find a maximum weight vertex $b$-packing in $\tilde{O}(m^{10/7}\log W)$
time.
\begin{proof}
Consider Algorithm \ref{algorithm:matching-scaling} that executes
cost scaling for the dual solution. It first recursively computes
the dual problem for a graph with costs that are smaller by a factor
of $2$. Then it uses the obtained solution as a starting point for
the computation of the dual solution in the original graph. 

\begin{algorithm}
\begin{enumerate}
\item If $c_{uv}=0$ for all $uv\in E$ then return $y_{v}=0$ for all $v\in V_{1}\cup V_{2}$.
\item Recursively find the optimal dual solution $\overline{y}$ for a graph
$\overline{G}=(V_{1}\cup V_{2},E,\overline{c})$ where $\overline{c}_{uv}=\lfloor\frac{c_{uv}}{2}\rfloor$
for all $uv\in E$.
\item Set $y_{v}=2\overline{y}_{v}$ for all $v\in V_{1}\cup V_{2}.$
\item Set $\tilde{c}_{uv}=c_{uv}-y_{u}-y_{v}$ for all $uv\in E$.
\item Set $\tilde{c}_{uv}^{n}=\min(\tilde{c}_{uv},\norm b_{1}).$
\item Find the optimal dual solution $\tilde{y}^{n}$ for the graph $\tilde{G}^{n}=(V_{1}\cup V_{2},E,\tilde{c}^{n})$.
\item Return $y+\tilde{y}^{n}$ as the optimal dual solution for $G$.
\end{enumerate}
\caption{Scaling algorithm for computing maximum weight b-packing in $G=(V_{1}\cup V_{2},E,c)$. }
\label{algorithm:matching-scaling}
\end{algorithm}

The only step of the algorithm that requires some explanation is step
5 -- if we would have removed this step the algorithm would compute
the dual solution as it only works with reduced weights. When this
step is present we just need to argue that $\tilde{y}^{n}$ is the
optimal dual solution for $\tilde{G}=(V_{1}\cup V_{2},E,\tilde{c})$.
This can be shown by arguing that the optimal primal solution in $\tilde{G}$
never uses edges of weight higher than $\norm b_{1}$. Take an optimal
primal solution $\tilde{M}$ for $\tilde{G}$ and take an optimal
primal solution $\overline{M}$ for $\overline{G}$. By the scaling
procedure edges in $\overline{M}$ have weights either $0$ or $1$,
as they were tight for $\overline{y}$. By optimality of $\tilde{M}$
we know that$\tilde{c}(\tilde{M})\le\tilde{c}(\overline{M})\le\norm b_{1},$
so no single edge of $\tilde{M}$ can have weight higher than $\norm b_{1}$. 
\end{proof}
\end{thm}
The above theorem can be easily extended to computing the minimum
weight perfect $b$-matching. We only need to take the optimal dual
solution and restrict our attention to edges which are tight with
respect to it. By finding any perfect $b$-matching in this unweighted
tight graph we obtain the minimum weight perfect $b$-matching in
the original graph.
\begin{cor}
\label{theorem:matchings-improved-1}Assuming $\norm b_{1}=O(m)$,
we can find a minimum weight perfect $b$-matching in $\tilde{O}(m^{10/7}\log W)$
time.
\end{cor}

\section{\label{sec:Shortest-Paths-with}Shortest Paths with Negative Weights}

We are given a directed graph $G(V,E,c)$ together with the edge weight
function $c:E\to\left\{ -W,\ldots,0,\ldots,W\right\} $ and a source
vertex $s$. Our goal is to compute shortest paths from $s$ to all
vertices in $V$. We will start by reducing this shortest paths problem
to the weighted perfect $1$-matching problem using the reduction
that was given by Gabow~\cite{gabow-85}. The main step of this reduction
is a construction of a bipartite graph $G_{12}=(V_{1}\cup V_{2},E_{12},c_{12})$
such that the vertex packing problem in $G_{12}$ induces a valid
potential function on $G$. Using this potential function we can reweigh
$w$ to remove negative weights.

We define a bipartite graph $G_{12}=(V_{1}\cup V_{2},E_{12},c_{12})$
in the following way 
\begin{eqnarray*}
V_{1} & = & \left\{ v_{1}\vert v\in V\right\} ,\\
V_{2} & = & \left\{ v_{2}\vert v\in V\right\} ,\\
E_{12} & = & \left\{ u_{2}v_{1}\vert uv\in E\right\} \cup\left\{ v_{1}v_{2}\vert v\in V\right\} .
\end{eqnarray*}
The weight function $w_{12}$ is defined as follows 
\begin{eqnarray*}
c_{12}(u_{i}v_{j}) & = & \left\{ \begin{array}{ll}
c_{uv} & \textnormal{if }uv\in E,\\
0 & \textnormal{otherwise}
\end{array}\right.
\end{eqnarray*}

Let us observe that a perfect $1$-matching in $G_{12}$ corresponds
to a set of cycles in the graph $G$. This leads to the following
observation.
\begin{lem}
\label{lemma:reduction-cycle} The graph $G$ contains a negative
length cycle if and only if the weight of minimum cost perfect $1$-matching
in $G_{12}$ is negative. 
\end{lem}
Hence, if the weight of minimum cost perfect $1$-matching in $G_{12}$
is negative we conclude that there is a negative weight cycle in the
graph $G$. Hence, we assume that the weight of minimum cost $1$-matching
is equal to $0$. Moreover, the dual solution to $1$-matchings, i.e.,
$1$-vertex packing in $G_{12}$ induces a potential function on $G$.
Let $y:V_{1}\cup V_{2}\to\mathcal{\mathbb{R}}$ be the maximum weight
vertex \textbf{$1$}-packing.
\begin{lem}
\label{lemma:reduction-potential} Let $y$ be maximum vertex $1$-packing
in $G_{12}$. If $y(V_{1}\cup V_{2})=0$ then $p_{v}:=y(v_{1})$ is
a potential function on $G$, i.e., $c_{uv}+p_{u}-p_{v}\ge0$. \end{lem}
\begin{proof}
As $y(V_{1}\cup V_{2})=c(M)=0$, i.e., the minimum weight perfect
matching in $G_{12}$ has weight zero. This means that all edges $v_{1}v_{2}$
for $v\in V$ are tight as they form a perfect $1$-matching of weight
$0$. This in turn implies that $y(v_{1})=-y(v_{2})$. Thus for an
edge $uv\in E$ we have 
\begin{eqnarray*}
y(u_{2})+y(v_{1}) & \le & c_{12}(u_{2}v_{1})\\
-y(u_{1})+y(v_{1}) & \le & c_{12}(u_{2}v_{1})\\
-p_{u}+p_{v} & \le & c_{uv}.
\end{eqnarray*}

\end{proof}
By combining Theorem\ref{theorem:matchings-dual} with the above lemma
we obtain the following, i.e., we use the potential function to obtain
a reweighed non-negative instance that can be solved using Dijkstra's
algorithm. 
\begin{cor}
\label{theorem:paths} Single source shortest paths in a graph with
negative weights can be computed in $\tilde{O}(m^{10/7}\log W)$ time.\end{cor}

\appendix

\section{\label{sec:app-proof-total-energy-incr}Proof of Lemma \ref{lem:total-energy-measure-increase-1}}
\begin{proof}
We require a simple refinement of the proof of Lemma \ref{lem:total-energy-measure-increase}.
The crucial aspect of this proof is that it does not require any control
over how the energy changes between two iterations belonging to different
phases.  Just as we did in (\ref{eq:energy-increase}), let us bound
the energy increase during a phase that starts at iteration $t_{0}+1$.
Let $\hat{T}=m^{2\eta}$ represent the length of a phase.

\begin{align*}
\sum_{t=t_{0}+1}^{t_{0}+\hat{T}}\Delta^{t} & \leq\sum_{t=t_{0}+1}^{t_{0}+\hat{T}}\left(\mathcal{E}^{t}-\mathcal{E}^{t-1}+5\crho^{2}\cdot m^{1-2\eta}\right)\\
 & \leq\mathcal{E}^{t_{0}+\hat{T}}+\hat{T}\cdot5\crho^{2}\cdot m^{1-2\eta}\\
 & \leq\norm{\nu^{t_{0}+\hat{T}}}_{1}+\hat{T}\cdot5\crho^{2}\cdot m^{1-2\eta}\\
 & \leq\left(\norm{\nu^{t_{0}}}_{1}+\sum_{t=t_{0}+1}^{t_{0}+\hat{T}}2C\cdot\Delta^{t}\right)+\hat{T}\cdot5\crho^{2}\cdot m^{1-2\eta}\\
 & =\norm{\nu^{t_{0}}}_{1}+2C\cdot\sum_{t=t_{0}+1}^{t_{0}+\hat{T}}\Delta^{t}+\hat{T}\cdot5\crho^{2}\cdot m^{1-2\eta}
\end{align*}

Therefore 
\begin{equation}
\sum_{t=t_{0}+1}^{t_{0}+\hat{T}}\Delta^{t}\leq\frac{1}{1-2C}\left(\norm{\nu^{t_{0}}}_{1}+\hat{T}\cdot5m^{1-2\eta}\right)=2\cdot\left(\norm{\nu^{t_{0}}}_{1}+5\crho^{2}\cdot m\right)\label{eq:phase-energy-increase}
\end{equation}

Also, the increase in measure during this phase satisfies 
\begin{equation}
\norm{\nu^{t_{0}+\hat{T}}}_{1}-\norm{\nu^{t_{0}}}_{1}\leq2C\cdot\left(\sum_{t=t_{0}+1}^{t_{0}+\hat{T}}\Delta^{t}\right)\leq4C\left(\norm{\nu^{t_{0}}}_{1}+5\crho^{2}\cdot m\right)=4C\norm{\nu^{t_{0}}}_{1}+m^{1/2+3\eta}/\cT\label{eq:measure-increase-phase}
\end{equation}

Using the fact that initially$\norm{\nu^{0}}=m$, and applying (\ref{eq:measure-increase-phase})
we obtain by induction that, for the first $K=\cT\cdot m^{1/2-3\eta}$
phases, the total measure increase during a phase satisfies

\[
\norm{\nu^{t_{0}+\hat{T}}}_{1}-\norm{\nu^{t_{0}}}_{1}\leq\frac{2}{\cT}\cdot m^{1/2+3\eta}
\]

This can be easily verified for the base case: $4Cm+m^{1/2+3\eta}/\cT=4\cdot\frac{m^{-1/2+3\eta}}{20\cT\crho^{2}}\cdot m+\frac{m^{1/2+3\eta}}{\cT}\leq\frac{2}{\cT}\cdot m^{1/2+3\eta}$.
For the induction step, we have by the induction hypothesis that after
$K$ phases $\norm{\nu^{t_{0}}}_{1}\leq m+K\cdot\frac{2}{\cT}\cdot m^{1/2+3\eta}$.
Therefore the measure increase in the new phase is at most $4C\left(m+K\cdot\frac{2}{\cT}\cdot m^{1/2+3\eta}\right)+m^{1/2+3\eta}/\cT\leq4\cdot\frac{m^{-1/2+3\eta}}{20\cT\crho^{2}}\cdot\left(m+2m\right)+\frac{m^{1/2+3\eta}}{\cT}\leq\frac{2}{\cT}\cdot m^{1/2+3\eta}$. 

Therefore, when the phases are over, we have

\begin{equation}
\norm{\nu}_{1}\leq m+\cT\cdot m^{1/2-3\eta}\cdot\frac{2}{\cT}\cdot m^{1/2+3\eta}\leq3m\label{eq:phase-final-measure}
\end{equation}

which shows that Invariant 1 holds.

Also, plugging in (\ref{eq:phase-energy-increase}) we obtain that
during a phase the total energy increase is at most 
\[
2\cdot\left(3m+5\crho^{2}\cdot m\right)\leq16\crho^{2}\cdot m
\]

\end{proof}

\section{\label{sec:orthog-flows}Proof of Lemma \ref{lem:orthog}}
\begin{proof}
Let $\elf$ be an electrical flow in the graph with auxiliary arcs,
and let $\pot$ be the corresponding vertex potentials. Let $\potp$
be a different set of vertex potentials satisfying

\[
\hat{\phi}_{v}^{(P)}=\begin{cases}
\pot_{v} & v\in P\\
\frac{r_{\bar{u}v}\pot_{u}+r_{uv}\pot_{\bar{u}}}{r_{uv}+r_{\bar{u}v}} & v\in Q,\ (u,v),(\bar{u},v)\in E\\
\pot_{v_{0}} & v=v_{0}
\end{cases}
\]

Let $\potq=\pot-\potp$. Also let $\elfp$ and $\elfq$ be the electrical
flows corresponding to potentials $\potp$, respectively $\potq$with
the same set of resistances $r$. So these flows can be constructed
directly from the vertex potentials $\pot$.

Now let us verify that $\elfp$ satisfies flow conservation on all
vertices in $Q\cup\left\{ v_{0}\right\} $. Given any vertex $v\in Q$,
let $u$ and $\bar{u}$ be the vertices corresponding to partner arcs
$(u,v)$ and $(\bar{u},v)$. By Ohm's law, we have 
\[
\hat{f}_{uv}^{(P)}=\frac{\hat{\phi}_{v}^{(P)}-\hat{\phi}_{u}^{(P)}}{r_{uv}}=\frac{\frac{r_{\bar{u}v}\pot_{u}+r_{uv}\pot_{\bar{u}}}{r_{uv}+r_{\bar{u}v}}-\hat{\phi}_{u}}{r_{uv}}=\frac{r_{uv}\left(\pot_{\bar{u}}-\hat{\phi}_{u}\right)}{r_{uv}\left(r_{uv}+r_{\bar{u}v}\right)}=\frac{\pot_{\bar{u}}-\hat{\phi}_{u}}{r_{uv}+r_{\bar{u}v}}
\]

and similarly

\[
\hat{f}_{\bar{u}v}^{(P)}=\frac{\hat{\phi}_{v}^{(P)}-\hat{\phi}_{\bar{u}}^{(P)}}{r_{\bar{u}v}}=\frac{\frac{r_{\bar{u}v}\pot_{u}+r_{uv}\pot_{\bar{u}}}{r_{uv}+r_{\bar{u}v}}-\hat{\phi}_{\bar{u}}}{r_{uv}}=\frac{r_{\bar{u}v}\left(\hat{\phi}_{u}-\pot_{\bar{u}}\right)}{r_{\bar{u}v}\left(r_{uv}+r_{\bar{u}v}\right)}=\frac{\hat{\phi}_{u}-\pot_{\bar{u}}}{r_{uv}+r_{\bar{u}v}}
\]

Hence $\hat{f}_{uv}^{(P)}+\hat{f}_{\bar{u}v}^{(P)}=0$, so flow is
conserved at $v$. The fact that flow is also conserved at $v_{0}$
is immediate, since potentials on $P\cup\left\{ v_{0}\right\} $ are
identical to those in $\hat{\phi}$, and the flow $\hat{f}$ corresponding
to these is always conserved at $v_{0}$.

Furthermore, since $f$ satisfies the same demands as $\elf$, $\fp$
also obeys flow conservation on $Q\cup\left\{ v_{0}\right\} $. Now,
consider restricting the flows $\fp$ and $\elfq$ to a pair of partner
edges, $e=(u,v)$ and $\bar{e}=(\bar{u},v)$ (i.e. zeroing out the
flows on all other edges). This restriction of $\elfq$ is an electrical
flow induced by a voltage vector nonzero only at $v$, so the $\mr$-inner
product with it simply measures net flow to $v$. As $\fp$ obeys
flow conservation at $v$, this implies that the restrictions of $\fp$
and $\elfq$ are $\mr$-orthogonal. That in turn implies that the
energy of each of these restricted flows is at most the energy of
the same restriction of $f$, which is equal to $\nu_{e}+\nu_{\bar{e}}.$
\end{proof}

\section{\label{sec:proof-progress-step}Proof of Theorem \ref{thm:l_4_improvement_step}}

Proving this theorem is done in two parts. In the first one we analyze
the predictor step, which produces a new iterate that is close to
the central path $\ell_{2}$ norm, and track the change in resistance.
In the second part we analyze the centering steps, which - starting
with a solution that is close to the central path - produce a centered
solution. In both cases, we will also show that the iterates stay
feasible at all times.

We start with analyzing the predictor step, described in Section \ref{par:Taking-an-Improvement}:

\begin{eqnarray*}
f_{e}' & := & (1-\delta)f_{e}+\delta\hat{f}_{e},\\
s_{e}' & := & s_{e}-\frac{\delta}{1-\delta}\left(\widehat{\phi}_{v}-\widehat{\phi}_{u}\right)
\end{eqnarray*}
for all arcs $e=(u,v)$.

We verify feasibility of $f'$ and $s'$, and bound the change in
resistance.
\begin{prop}
\label{prop:resistances-predictor}Both $f'$ and $s'$ are feasible,
i.e. $f',s'>0$. Furthermore, $\frac{r_{e}}{r_{e}'}=\frac{s_{e}}{f_{e}}\cdot\frac{f_{e}'}{s_{e}'}\leq1+4\delta\rho_{e}$.\end{prop}
\begin{proof}
We obtain 
\begin{align*}
f_{e'} & =(1-\delta)f_{e}+\delta\hat{f}_{e}\geq(1-\delta)f_{e}-\delta\abs{\hat{f}_{e}}=\left(1-\delta-\delta\rho_{e}\right)f_{e}\geq\left(1-\delta-\frac{1}{8\norm{\rho}_{\nu,4}}\rho_{e}\right)f_{e}\\
 & \geq\left(1-\frac{1}{8}-\frac{1}{8\norm{\rho}_{\infty}}\rho_{e}\right)f_{e}\geq\left(1-\frac{1}{8}-\frac{1}{8}\right)f_{e}>0\\
s_{e'} & =s_{e}-\frac{\delta}{1-\delta}\left(\widehat{\phi}_{v}-\widehat{\phi}_{u}\right)=s_{e}-\frac{\delta}{1-\delta}\cdot\frac{s_{e}}{f_{e}}\cdot\hat{f}_{e}=s_{e}\left(1-\frac{\delta}{1-\delta}\cdot\frac{\hat{f}_{e}}{f_{e}}\right)\\
 & \geq s_{e}\left(1-\frac{\delta}{1-\delta}\cdot\rho_{e}\right)\geq s_{e}\left(1-\delta\cdot\rho_{e}\right)\geq s_{e}\left(1-\frac{1}{8\norm{\rho}_{\nu,4}}\rho_{e}\right)\\
 & \geq s_{e}\left(1-\frac{1}{8\norm{\rho}_{\infty}}\rho_{e}\right)\geq s_{e}\left(1-\frac{1}{8}\right)>0
\end{align*}

Next, we analyze the change in resistance:

\begin{align*}
\frac{r_{e}}{r_{e}'} & =\frac{s_{e}}{f_{e}}\cdot\frac{f_{e}'}{s_{e}'}=\frac{s_{e}}{f_{e}}\cdot\frac{(1-\delta)f_{e}+\delta\hat{f}_{e}}{s_{e}\left(1-\frac{\delta}{1-\delta}\cdot\frac{\hat{f}_{e}}{f_{e}}\right)}=\frac{1-\delta+\delta\frac{\hat{f}_{e}}{f_{e}}}{1-\frac{\delta}{1-\delta}\cdot\frac{\hat{f}_{e}}{f_{e}}}\leq(1-\delta)\frac{1-\delta+\delta\rho_{e}}{1-\delta-\delta\rho_{e}}\\
 & \leq\frac{1-\delta+\delta\rho_{e}}{1-\delta-\delta\rho_{e}}=1+\frac{2\delta\rho_{e}}{1-\delta-\delta\rho_{e}}\leq1+\frac{2\delta\rho_{e}}{1-\frac{1}{8}-\frac{1}{8\norm{\rho}_{\nu,4}}\cdot\rho_{e}}\\
 & \leq1+\frac{2\delta\rho_{e}}{1-\frac{1}{8}-\frac{1}{8\norm{\rho}_{\infty}}\cdot\rho_{e}}\leq1+\frac{2\delta\rho_{e}}{1-\frac{1}{8}-\frac{1}{8}}\leq1+4\delta\rho_{e}
\end{align*}

\end{proof}
As we saw in Section \ref{par:Congestion-Vector.}, taking the predictor
step produces a solution that is close to the central path in $\ell_{2}$
norm. We will now show that once this happens, only a few centering
steps are required.

For simplicity, we will overload notation for the rest of the section,
in order to measure centrality with respect to the following quantity:
\begin{defn}
$\norm{\frac{fs}{\mu}-1}_{\nu,2}:=\sqrt{\sum_{e}\nu_{e}\left(\frac{f_{e}s_{e}}{\mu_{e}}-1\right)^{2}}$
\end{defn}
Now let us describe a centering step. Given an instance $(f,s,\nu)$
such that $\norm{\frac{fs}{\mu}-1}_{\nu,2}^{2}=\sum_{e}\nu_{e}\left(\frac{f_{e}s_{e}}{\mu_{e}}-1\right)^{2}\leq\frac{1}{256}$,
we show how to produce a new instance with better centrality. In order
to do so, consider a flow $f^{\#}$ such that $f_{e}^{\sharp}=\frac{\mu_{e}}{s_{e}}$.
This is clearly central, but it routes a different demand. In order
to produce a flow that routes the correct demand, we take the demand
of $f-f^{\sharp}$ and route it electrically with resistances $\tilde{r}_{e}=\frac{s_{e}}{f_{e}^{\sharp}}=\frac{s_{e}^{2}}{\mu_{e}}$
producing an electrical flow $\tilde{f}$. Then we set the new flow
\begin{align*}
f_{e}' & =f_{e}^{\sharp}+\tilde{f}_{e}\\
s_{e}' & =s_{e}-\frac{s_{e}}{f_{e}^{\sharp}}\cdot\tilde{f}_{e}
\end{align*}

We first verify that these flows stay feasible. In order to do so,
we will require a proposition that we will use in upper bounding the
ratio $\abs{\frac{\tilde{f}_{e}}{f_{e}^{\sharp}}}$.
\begin{prop}
\label{prop:bound-flow-ratio}$\norm{\frac{\tilde{f}_{e}}{f_{e}^{\sharp}}}_{\nu,2}\leq\norm{\frac{fs}{\mu}-1}_{\nu,2}$\end{prop}
\begin{proof}
This is a straightforward energy minimization argument. We first express
$\norm{\frac{\tilde{f}}{f^{\sharp}}}_{\nu,2}^{2}$ in terms of the
energy of the electrical flow $\tilde{f}$. Then, using the fact that
$f-f^{\sharp}$ routes the same demand as the electrical flow $\tilde{f}$,
the energy of $f-f^{\sharp}$ can only be larger than that of $\tilde{f}$.
This yields:\\
\begin{align*}
\norm{\frac{\tilde{f}}{f^{\sharp}}}_{\nu,2}^{2} & =\sum_{e}\nu_{e}\left(\frac{\tilde{f}_{e}}{f_{e}^{\sharp}}\right)^{2}=\sum_{e}\frac{\nu_{e}}{s_{e}f_{e}^{\sharp}}\cdot\tilde{r}_{e}\tilde{f}_{e}^{2}=\sum_{e}\frac{\nu_{e}}{\mu_{e}}\cdot\tilde{r}_{e}\tilde{f}_{e}^{2}=\frac{1}{\hat{\mu}}\cdot\sum_{e}\tilde{r}_{e}\tilde{f}_{e}^{2}\\
 & \leq\frac{1}{\hat{\mu}}\sum_{e}\tilde{r}_{e}\left(f_{e}-f_{e}^{\sharp}\right)^{2}=\sum_{e}\frac{\nu_{e}}{\left(f_{e}^{\sharp}\right)^{2}}\cdot(f_{e}-f_{e}^{\sharp})^{2}\\
 & =\sum_{e}\nu_{e}\left(\frac{f_{e}}{f_{e}^{\sharp}}-1\right)^{2}=\sum_{e}\nu_{e}\left(\frac{f_{e}s_{e}}{\mu_{e}}-1\right)^{2}\\
 & =\norm{\frac{fs}{\mu}-1}_{\nu.2}^{2}
\end{align*}
 
\end{proof}
Now we can prove the following proposition:
\begin{prop}
\label{prop:resistances-corrector}After taking a centering step,
$f'$ and $s'$ stay feasible, i.e. $f'$,$s'>0$. Furthermore $\norm{\frac{f's'}{\mu}-1}_{\nu,2}\leq\norm{\frac{fs}{\mu}-1}_{\nu,2}^{2}$.
Also, $\frac{r_{e}}{r_{e'}}=\frac{s_{e}}{f_{e}}\cdot\frac{f_{e}'}{s_{e}'}\leq\left(1+2\abs{\frac{s_{e}f_{e}}{\mu}-1}\right)\left(1+4\abs{\frac{\tilde{f}_{e}}{f_{e}^{\sharp}}}\right)$.\end{prop}
\begin{proof}
To lower bound $f_{e}'$, we can simply use Proposition \ref{prop:bound-flow-ratio}
to upper bound $\abs{\frac{\tilde{f}_{e}}{f_{e}^{\sharp}}}\leq\norm{\frac{fs}{\mu}-1}_{\nu,2}\leq\frac{1}{16}$.
This yields $\abs{\tilde{f}_{e}}\leq f_{e}^{\sharp}\cdot\frac{1}{16}$.
Plugging in, we obtain:

\begin{align*}
f_{e}' & =f_{e}^{\sharp}+\tilde{f}_{e}\geq f_{e}^{\sharp}-f_{e}^{\sharp}\cdot\norm{\frac{fs}{\mu}-1}_{\nu,2}\geq f_{e}^{\sharp}-f_{e}^{\sharp}\cdot\frac{1}{16}>0
\end{align*}

In order to prove feasibility for the new slack iterate, we observe
that

\begin{align*}
s_{e}' & =s_{e}\left(1-\frac{\tilde{f}_{e}}{f_{e}^{\sharp}}\right)\geq s_{e}\left(1-\frac{1}{16}\right)>0
\end{align*}

Then we verify that this step improves centrality.

\begin{align*}
\sum_{e}\nu_{e}\left(\frac{f_{e}'s_{e}'}{\mu_{e}}-1\right)^{2} & =\sum_{e}\nu_{e}\left(\frac{\left(f_{e}^{\sharp}+\tilde{f}_{e}\right)\left(s_{e}-\frac{s_{e}}{f_{e}^{\sharp}}\cdot\tilde{f}_{e}\right)}{\mu_{e}}-1\right)^{2}=\sum_{e}\nu_{e}\left(\frac{f_{e}^{\sharp}s_{e}-\frac{s_{e}}{f_{e}^{\sharp}}\cdot\tilde{f}_{e}^{2}}{\mu_{e}}-1\right)^{2}\\
 & =\sum_{e}\nu_{e}\left(\frac{\mu_{e}-\mu_{e}\left(\frac{\tilde{f}_{e}}{f_{e}^{\sharp}}\right)^{2}}{\mu_{e}}-1\right)^{2}=\sum_{e}\nu_{e}\left(\frac{\tilde{f}_{e}}{f_{e}^{\sharp}}\right)^{4}
\end{align*}

Applying Cauchy-Schwarz, we see that this is upper bounded by 
\begin{eqnarray*}
\left(\sum_{e}\nu_{e}\cdot\left(\frac{\tilde{f}_{e}}{f_{e}^{\sharp}}\right)^{2}\right)^{2} & = & \left(\sum_{e}\nu_{e}\cdot\tilde{r}_{e}\frac{f_{e}^{\sharp}}{s_{e}}\cdot\left(\frac{\tilde{f}_{e}}{f_{e}^{\sharp}}\right)^{2}\right)^{2}=\left(\sum_{e}\nu_{e}\cdot\frac{\tilde{r}_{e}}{\mu_{e}}\cdot\tilde{f}_{e}^{2}\right)^{2}\\
 & = & \left(\sum_{e}\nu_{e}\cdot\frac{\tilde{r}_{e}}{\hat{\mu}\cdot\nu_{e}}\cdot\tilde{f}_{e}^{2}\right)^{2}=\left(\frac{1}{\hat{\mu}}\cdot\sum_{e}\tilde{r}_{e}\tilde{f}_{e}^{2}\right)^{2}
\end{eqnarray*}

Now note that the last term contains the energy of the electrical
flow $\tilde{f}$ with respect to resistances $\tilde{r}$. Since
$\tilde{f}$ and $f-f^{\sharp}$ route the same demand, but $\tilde{f}$
is an electrical flow, so it minimizes energy, we know that this is
upper bounded by 
\begin{eqnarray*}
\left(\frac{1}{\hat{\mu}}\cdot\sum_{e}\tilde{r}_{e}\left(f_{e}-f_{e}^{\sharp}\right)^{2}\right)^{2} & = & \left(\frac{1}{\hat{\mu}}\cdot\sum_{e}\frac{s_{e}^{2}}{\mu_{e}}\left(f_{e}-\frac{\mu_{e}}{s_{e}}\right)^{2}\right)^{2}\\
 & = & \left(\frac{1}{\hat{\mu}}\cdot\sum_{e}\mu_{e}\left(\frac{f_{e}s_{e}}{\mu_{e}}-1\right)^{2}\right)^{2}\\
 & = & \left(\sum_{e}\nu_{e}\left(\frac{f_{e}s_{e}}{\mu_{e}}-1\right)^{2}\right)^{2}
\end{eqnarray*}

This shows that during every such step, centrality improves at a quadratic
rate.

Next we bound how much resistances get changed by a single centering
step:

\begin{align*}
\frac{r_{e}}{r_{e}'} & =\frac{s_{e}}{f_{e}}\cdot\frac{f_{e}'}{s_{e}'}=\frac{s_{e}}{f_{e}}\cdot\frac{f_{e}^{\sharp}+\tilde{f}_{e}}{s_{e}\left(1-\frac{\tilde{f}_{e}}{f_{e}^{\sharp}}\right)}=\frac{f_{e}^{\sharp}}{f_{e}}\cdot\frac{1+\frac{\tilde{f}_{e}}{f_{e}^{\sharp}}}{1-\frac{\tilde{f}_{e}}{f_{e}^{\sharp}}}=\frac{\mu_{e}}{s_{e}f_{e}}\cdot\frac{1+\frac{\tilde{f}_{e}}{f_{e}^{\sharp}}}{1-\frac{\tilde{f}_{e}}{f_{e}^{\sharp}}}\leq\frac{\mu_{e}}{s_{e}f_{e}}\cdot\left(1+2\abs{\frac{\tilde{f}_{e}}{f_{e}^{\sharp}}}\right)\\
 & \leq\left(1+2\abs{\frac{s_{e}f_{e}}{\mu_{e}}-1}\right)\left(1+4\abs{\frac{\tilde{f}_{e}}{f_{e}^{\sharp}}}\right)
\end{align*}

We used the inequalities $\frac{1}{1+x}\leq1+2\abs x$, and $\frac{1+x}{1-x}\leq1+4\abs x$,
for $\abs x\leq1/2$.
\end{proof}
Now we can track how resistances change after restoring to near-perfect
centrality. This is highlighted by the following proposition.
\begin{prop}
After a predictor step, followed by any number of centering steps,
the new set of resistances $r^{(k)}$ satisfy $\frac{r_{e}}{r_{e}^{(k)}}\leq1+4\delta\rho_{e}+\kappa_{e}$,
for some $\kappa$ such that $\norm{\kappa}_{\nu,2}\leq1$.\end{prop}
\begin{proof}
Let $f^{(i)}$ and $s^{(i)}$ be the flow and slack iterates obtained
after applying a predictor step to $(f,s,\nu)$, followed by $i$
centering steps. Let $\tilde{f}^{(i)}$ and $f^{\sharp(i)}$ be defined
similarly. Combining Proposition \ref{prop:resistances-predictor}
with Proposition \ref{prop:resistances-corrector} we can upper bound
the resistance $r^{(k)}=\frac{s^{(k)}}{f^{(k)}}$:
\end{proof}
\begin{align*}
\frac{r_{e}}{r_{e}^{(k)}} & \leq(1+4\delta\rho_{e})\cdot\prod_{i=0}^{k-1}\left(1+2\abs{\frac{f_{e}^{(i)}s_{e}^{(i)}}{\mu}-1}\right)\left(1+4\abs{\frac{\tilde{f}_{e}^{(i)}}{f_{e}^{\sharp(i)}}}\right)\\
 & \leq(1+4\delta\rho_{e})\cdot\exp\left(2\left(\sum_{i=0}^{k-1}2\abs{\frac{f_{e}^{(i)}s_{e}^{(i)}}{\mu}-1}+4\abs{\frac{\tilde{f}_{e}^{(i)}}{f_{e}^{\sharp(i)}}}\right)\right)\\
 & \leq(1+4\delta\rho_{e})\cdot\left(1+\left(\sum_{i=0}^{k-1}8\abs{\frac{f_{e}^{(i)}s_{e}^{(i)}}{\mu}-1}+16\abs{\frac{\tilde{f}_{e}^{(i)}}{f_{e}^{\sharp(i)}}}\right)\right)\\
 & \leq1+4\delta\rho_{e}+\frac{1}{2}\cdot\sum_{i=0}^{k-1}\left(8\abs{\frac{f_{e}^{(i)}s_{e}^{(i)}}{\mu}-1}+16\abs{\frac{\tilde{f}_{e}^{(i)}}{f_{e}^{\sharp(i)}}}\right)\\
 & =1+4\delta\rho_{e}+\sum_{i=0}^{k-1}\left(4\abs{\frac{f_{e}^{(i)}s_{e}^{(i)}}{\mu}-1}+8\abs{\frac{\tilde{f}_{e}^{(i)}}{f_{e}^{\sharp(i)}}}\right)\\
 & =1+4\delta\rho_{e}+\kappa_{e}
\end{align*}

We used the fact that $\sum_{i=0}^{k-1}2\abs{\frac{f_{e}^{(i)}s_{e}^{(i)}}{\mu}-1}+4\abs{\frac{\tilde{f}_{e}^{(i)}}{f_{e}^{\sharp(i)}}}\leq\sum_{i=0}^{k-1}6\norm{\frac{f^{(i)}s^{(i)}}{\mu}-1}_{\nu,2}\leq6\cdot\frac{1}{16}\leq1$,
and $\exp(x)\leq1+2\abs x$ for $\abs x\leq1$.

Now to bound the norm $\norm{\kappa}_{\nu,2}$, we use triangle inequality
and bound the contribution of each centering step separately:
\begin{align*}
\norm{\kappa}_{\nu,2} & \leq\sum_{i=0}^{k-1}\left(4\norm{\frac{f^{(i)}s^{(i)}}{\mu}-1}_{\nu,2}+8\norm{\frac{\tilde{f}^{(i)}}{f^{\sharp(i)}}}_{\nu,2}\right)\\
 & \leq12\sum_{i=0}^{k-1}\norm{\frac{f^{(i)}s^{(i)}}{\mu}-1}_{\nu,2}\\
 & \leq12\sum_{i=0}^{k-1}\left(\frac{1}{16}\right)^{2^{i}}\\
 & \leq12\cdot\frac{1}{15}\\
 & \le1
\end{align*}

\begin{rem}
In Proposition \ref{prop:resistances-corrector} we saw that centrality
decreases at a quadratic rate each iteration. Therefore centrality
decreases to machine precision in $\tilde{O}(\log\log U)$ iterations.
For this reason, we perform centering steps until we restore centrality
to machine precision $O(\epsilon_{\textnormal{mach}})$ (the number
of such iterations will be absorbed by the $\tilde{O}$ in the running
time). Although throughout the paper we work with exact centrality,
rather than having it set to $O(\epsilon_{\textnormal{mach}})$, fixing
this can easily be done by perturbing the measures at the end of the
progress step in order to guarantee exact centrality. This changes
all measures by $O(\epsilon_{\textnormal{mach}})$; all the analysis
we do is still valid, but this extra change needs to be carried over
throughout the computations. We chose to ignore it in order to simplify
the presentation.\end{rem}

\section{\label{sec:precon-proofs}Preconditioning proofs.}

The purpose of this section is to prove the core result, Lemma \ref{lem:full-precon},
on the guarantees of preconditioning. However, we begin with a simpler,
more general statement:
\begin{lem}
\label{lem:partial-precon}Let $N$ be an electrical network, let
$\fn$ be a flow on $N$, and let $\elfn$ be the electrical flow
on $N$ satisfying the same demands as $\fn$, induced by voltages
$\potn$. Let $\phi_{0}$ be any fixed real number (representing an
absolute voltage). Then
\[
\sum_{\substack{e=(u,v),\\
\max(|\potn_{u}-\phi_{0}|,|\potn_{v}-\phi_{0}|)\geq V/4
}
}r_{e}(\fn_{e})^{2}\geq\frac{1}{2}\sum_{\substack{e=(u,v),\\
|\potn_{u}-\potn_{v}|\geq V
}
}r_{e}(\elfn_{e})^{2}.
\]
\end{lem}
\begin{proof}
Given a set of vertex potentials $\potn$, a flow $f$ and an absolute
voltage $\phi$, we define
\[
F(\potn,f,\phi)=\sum_{\substack{e=(u,v)\\
\min(\potn_{u},\potn_{v})\leq\phi\leq\max(\potn_{u},\potn_{v})
}
}|f_{e}|.
\]

That is, $F(\potn,f,\phi)$ is the sum of the absolute values of the
flows of edges crossing the cut induced by the voltage $\phi$. Notably,
$\int_{-\infty}^{\infty}F(\potn,\elfn,\phi)\;d\phi$ is equal to the
energy of the electrical flow $\elfn$, since the contribution of
each edge $e=(u,v)$ is $|\elfn_{e}||\potn_{v}-\potn_{u}|=r_{e}|\elfn_{e}|^{2}$.
Furthermore, for all $\phi$ we have $F(\potn,\elfn,\phi)\leq F(\potn,\fn,\phi)$.
This is because $\elfn$ satisfies the same demands as $\fn$, and
thus they have the same \emph{net} flow across each cut, including
the cuts induced by the voltage $\phi$. In the electrical flow, all
of the flow is in the direction of increasing $\potn$ by definition;
thus the sum of the absolute values of the edge flows, $F(\potn,\elfn,\phi)$,
is equal to that net flow value. In $\fn$, on the other hand, that
net flow value is a lower bound for $F(\potn,\fn,\phi)$.

We now define
\begin{align*}
X & =\left(\int_{-\infty}^{\phi_{0}-\frac{V}{4}}F(\potn,\elfn,\phi)\;d\phi\right)+\left(\int_{\phi_{0+\frac{V}{4}}}^{\infty}F(\potn,\elfn,\phi)\;d\phi\right)\\
Y & =\left(\int_{-\infty}^{\phi_{0}-\frac{V}{4}}F(\potn,\fn,\phi)\;d\phi\right)+\left(\int_{\phi_{0+\frac{V}{4}}}^{\infty}F(\potn,\fn,\phi)\;d\phi\right).
\end{align*}

Because the inequality holds for the integrands everywhere, we have
that $X\leq Y$. Alternatively, we may define $t_{e}$ for an edge
$(u,v)$ to be the fraction of the interval $[\min(\potn_{u},\potn_{v}),\max(\potn_{u},\potn_{v})]$
that is not contained in the interval $[\phi_{0}-\frac{V}{4},\phi_{0}+\frac{V}{4}]$.
We can then expand out $X$ and $Y$ as sums over edges, and define
a new $Z$:
\begin{align*}
X & =\sum_{e}t_{e}\abs{\elfn_{e}}\abs{\potn_{v}-\potn_{u}}\\
 & =\sum_{e}t_{e}r_{e}\abs{\elfn_{e}}^{2}\\
Y & =\sum_{e}t_{e}\abs{\fn_{e}}\abs{\potn_{v}-\potn_{u}}\\
 & =\sum_{e}t_{e}r_{e}\abs{\elfn_{e}}\abs{\fn_{e}}\\
Z & =\sum_{e}t_{e}r_{e}\abs{\fn_{e}}^{2}.
\end{align*}

Then by Cauchy-Schwarz on the sum expressions, we have $Y\leq\sqrt{XZ}$;
on the other hand, since $X\leq Y$ this means that we have $X\leq\sqrt{XZ}$
and so $Z\geq X$. Finally,
\[
\sum_{\substack{e=(u,v),\\
\max(|\potn_{u}-\phi_{0}|,|\potn_{v}-\phi_{0}|)\geq V/4
}
}r_{e}(\fn_{e})^{2}\geq Z
\]

since only those edges with an endpoint outside of $[\phi_{0}-\frac{V}{4},\phi_{0}+\frac{V}{4}]$
have $t_{e}>0$, while we have
\[
X\geq\frac{1}{2}\sum_{\substack{e=(u,v),\\
|\potn_{u}-\potn_{v}|\geq V
}
}r_{e}(\elfn_{e})^{2}
\]

since each of the edges in the latter sum, with their interval of
length $\abs{\potn_{u}-\potn_{v}}\geq V$, has $t_{e}\geq\frac{1}{2}$.
\end{proof}
We can now prove Lemma \ref{lem:full-precon}:
\begin{proof}[Proof of Lemma \ref{lem:full-precon}]
Apply Lemma \ref{lem:partial-precon} to $N$ and $\fn$, with $\phi_{0}=\potn_{v_{0}}$.
Now, note that since $r_{e}(\fn_{e})^{2}\leq\nu_{e}$, the total energy
of $\fn$, $\sum_{e}r_{e}(\fn_{e})^{2}$, is at most $\sum_{e}\nu_{e}$.
As an electrical flow satisfying the same demands, the total energy
of $\elfn$ is also at most $\sum_{e}\nu_{e}$.

On the other hand, every node $v$ is connected to $v_{0}$ by an
edge with resistance $\frac{R}{a(v)}$, which contributes $\frac{(\potn_{v}-\potn_{v_{0}})^{2}a(v)}{R}$.
Thus the total energy is at least
\begin{align*}
\sum_{v}\frac{(\potn_{v}-\potn_{v_{0}})^{2}a(v)}{R} & =\sum_{v}\frac{(\potn_{v}-\phi_{0})^{2}a(v)}{R}\\
 & \geq\frac{V^{2}}{16R}\sum_{v,|\potn_{v}-\phi_{0}|\geq\frac{V}{4}}a(v)\\
 & \geq\frac{V^{2}}{16R}\sum_{\substack{e=(u,v),\\
\max(|\potn_{u}-\phi_{0}|,|\potn_{v}-\phi_{0}|)\geq\frac{V}{4}
}
}\nu'_{e}\\
 & \geq\frac{V^{2}}{16R}\left(\left(\sum_{\substack{e=(u,v),\\
\max(|\potn_{u}-\phi_{0}|,|\potn_{v}-\phi_{0}|)\geq\frac{V}{4}
}
}\nu{}_{e}\right)-\sum_{e}(\nu_{e}-\nu'_{e})\right)\\
 & \geq\frac{V^{2}}{16R}\left(\left(\sum_{\substack{e=(u,v),\\
\max(|\potn_{u}-\phi_{0}|,|\potn_{v}-\phi_{0}|)\geq\frac{V}{4}
}
}r_{e}(\fn_{e})^{2}\right)-M\right)\\
 & \geq\frac{V^{2}}{16R}\left(\frac{1}{2}\left(\sum_{\substack{e=(u,v),\\
|\potn_{u}-\potn_{v}|\geq V
}
}r_{e}(\elfn_{e})^{2}\right)-M\right)\\
 & =\frac{V^{2}}{16R}\left(\frac{1}{2}\left(\sum_{r_{e}\abs{\elfn_{e}}\geq V}r_{e}(\elfn_{e})^{2}\right)-M\right).
\end{align*}

We therefore have
\begin{align*}
\frac{V^{2}}{16R}\left(\frac{1}{2}\left(\sum_{r_{e}\abs{\elfn_{e}}\geq V}r_{e}(\elfn_{e})^{2}\right)-M\right) & \leq\sum_{e}\nu_{e}\\
\sum_{r_{e}\abs{\elfn_{e}}\geq V}r_{e}(\elfn_{e})^{2} & \leq\frac{32R\sum_{e}\nu_{e}}{V^{2}}+2M.
\end{align*}

This completes the proof.\end{proof}

\bibliographystyle{plain}
\bibliography{ref}

\end{document}